\documentclass[onecolumn,pra, superscriptaddress,longbibliography,margin=1mm 11pt]{revtex4-2} 

\usepackage{graphicx}
\usepackage{bm}
\usepackage{amsmath}
\usepackage{environ}
\usepackage{appendix}
\usepackage{physics}

\usepackage[nopar]{lipsum}

\usepackage{amsthm}
\usepackage{amssymb}
\usepackage{amsfonts}
\usepackage{fancyhdr}
\usepackage{slashed}
\usepackage{bbm}
\usepackage{latexsym,epsfig,bbm}
\usepackage[colorlinks=true,urlcolor=blue,citecolor=blue]{hyperref}
\newcommand{\littleo}{{\scriptstyle \mathcal{O}}}
\usepackage{color}
\definecolor{blue}{rgb}{0,0.2,1}

\definecolor{red}{rgb}{0.9,0,0}

\newcommand{\vect}[1]{\boldsymbol{#1}}

\newtheorem{theorem}{Theorem}
\newtheorem{lemma}[theorem]{Lemma}
\newtheorem{definition}[theorem]{Definition}
\newtheorem{problem}{Problem}
\newtheorem{corollary}[theorem]{Corollary}

\begin{document}

\title{Quantum algorithms for uncertainty quantification: application to partial differential equations}

\date{\today}

\author{Francois Golse}
\affiliation{Ecole Polytechnique, CMLS, 91128 Palaiseau Cedex,
France.}  

\author{Shi Jin}
\affiliation{Institute of Natural Sciences, Shanghai Jiao Tong University, Shanghai 200240, China}
\affiliation{School of Mathematical Sciences, Shanghai Jiao Tong University, Shanghai 200240, China}
\affiliation{Ministry of Education Key Laboratory in Scientific and Engineering Computing, Shanghai Jiao Tong University, Shanghai 200240, China}

\author{Nana Liu}
\email{Corresponding author:  nana.liu@quantumlah.org}
\affiliation{Institute of Natural Sciences, Shanghai Jiao Tong University, Shanghai 200240, China}
\affiliation{Ministry of Education Key Laboratory in Scientific and Engineering Computing, Shanghai Jiao Tong University, Shanghai 200240, China}
\affiliation{University of Michigan-Shanghai Jiao Tong University Joint Institute, Shanghai 200240, China.}

\begin{abstract}
 Most problems in uncertainty quantification, despite its ubiquitousness in scientific computing, applied mathematics and data science, remain formidable on a classical computer. For uncertainties that arise in partial differential equations (PDEs), large numbers $M>>1$ of samples are required to obtain accurate ensemble averages. This usually involves solving the PDE $M$ times. In addition, to characterise the stochasticity in a PDE, the dimension $L$ of  the random input variables is high in most cases, and classical algorithms suffer from curse-of-dimensionality. We propose new quantum algorithms for PDEs with uncertain coefficients that are more efficient in $M$ and $L$ in various important regimes, compared to their classical counterparts. We introduce  transformations that transfer the original $d$-dimensional equation (with uncertain coefficients) into $d+L$ (for dissipative equations)  or $d+2L$ (for wave type equations) dimensional equations (with certain coefficients) in which the uncertainties appear only in the initial data. These transformations also allow one to superimpose the $M$ different initial data, so the  computational cost for the quantum algorithm to obtain the ensemble average from $M$ different samples is then {\it independent} of $M$, while also showing potential advantage in $d$, $L$ and precision $\epsilon$ in computing ensemble averaged solutions or physical observables. 
\end{abstract}
\maketitle 


\section{Introduction}

Most physical models  are not based on first-principles and are 
thus subject to uncertainties due to modeling or measurement errors. Examples include constitutive relations or equations of state in continuum mechanics, scattering coefficients in wave propagation and transport, initial or boundary data, forcing or source terms, diffusivity in porous or heterogeneous media.  Quantifying these uncertainties are important to validate, verify and calibrate the models, to conduct risk management, and to control the uncertainties \cite{smith2013, ghanem2017}. In the last two decades the area of uncertainty quantification (UQ) was one of the most active areas of research in scientific computing, applied mathematics and data sciences. \\

Uncertainties are typically modelled stochastically. In the case of ordinary and partial differential equations (ODEs and PDEs), one often uses stochastic ODEs and PDEs. To solve such models the frequently-used method is the Monte-Carlo method, which does not suffer from the curse-of-dimensionality. Its drawback, however, is its low--in fact only halfth- order accuracy and statistical noise. Therefore a large number of samples $M$ is needed for simulation so their ensemble averages are accurate.  Higher order methods have been introduced more recently, for example the stochastic Galerkin and stochastic collocation methods \cite{xiu2010, GWZ-Acta}, that are attractive when the solution has good regularity in the random space. \\

In most cases the dimension of the random space is high. For example, the stochastic process is often approximated by the Karhunen-Loeve expansion \cite{GS03}, which is a linear combination of a large number of i.i.d. random variables (referred to as random inputs). Hence the dimension of the random input variables could be high and one needs to needs to solve the underlying uncertain differential equations in very high-dimensional parameter space, on top of the possibly already high dimensionality in the physical space. Sparse grid \cite{bungartz2004} or greedy algorithms \cite{cohen2015} have been used but they are not effective in high dimensions. They also need sufficient regularity requirements on the parameter space and the accuracy could still be just marginally higher than the Monte-Carlo methods. Therefore, most of UQ problems remain formidable for a classical computer. \\

On the other hand, quantum algorithms based on computations on qubits, could offer possibly up to exponential speedup over their classical counterparts, although these are also subject to many caveats. While quantum computers that can solve practical problems are still out of reach in the near future, the development of quantum algorithms that show quantum advantages over classical algorithms have seen increasing activities in recent decades, for instance in linear algebra problems \cite{HHL-2009, Childs-2017, qsvd}.  Quantum algorithms for certain ODEs and PDEs (linear and nonlinear) have been proposed which, in certain regimes, could also demonstrate quantum advantages in dimension, precision, and the number of simulations
\cite{Berry-2014, joseph2020, dodin2021, lloyd2020quantum, Childs-2021, Liu2021nonlinear, jin2022time, jin-liu-2022}. Here one usually discretizes the equations in space and time, and then formulate them into linear algebra problems which are then solved by quantum algorithms for linear systems of equations. \\

In this paper we propose quantum algorithms for some of the most important linear PDEs with uncertain coefficients. We aim to develop efficient quantum algorithms that allow a {\it large} number of simulations for {\it multiple} initial and/or boundary data, in particular those arising in sampling methods, such as Monte-Carlo methods or stochastic collocation methods for uncertain PDEs. In these simulations one has to run repeated simulations for a large number of samples, or initial and boundary data in order to obtain  ensemble averages of the solutions. \\

Our main idea is to use  transformations that transfer the equation to {\it phase} space, so the 
random coefficients appear {\it only in the initial  data}, not in the equation itself. This allows us to run the simulation with  multiple initial data just {\it once} to get the ensemble average of the solution. This means that we gain speedup with respect to $M$ -- the (possibly very large) number of samples or initial data. While the original transformation was introduced in \cite{BGS2010}, it was used for the homogenization of PDEs with oscillatory coefficients. Here we use it for a completely different purpose and we also generalize the transformation for our purpose in the case of more general heterogeneous uncertain coefficients. The resulting phase space equations are then defined in higher dimensions. This means it can no longer be efficiently solved on a classical device, but can be efficiently solved by quantum linear PDE solvers with quantum advantage with respect to several critical
parameters including physical dimension. Most importantly, by transferring the stochastic parameters into the initial  conditions, one only needs to solve a {\it single} instance of the PDE to arrive at the ensemble average of the solutions or quantities of interest.  \\

Throughout this paper, when we say quantum advantage, we refer to the comparison to canonical classical algorithms that are used in practice --  which have good numerical accuracy and high resolution -- to solve these problems. Here we use $r^{\text{th}}$-order finite difference approximations. \\

We will use some of the most important linear PDEs to illustrate our approach, and then compare with the cost of both classical algorithms 
and standard quantum algorithms that simulate directly the original problems to demonstrate the new advantage gained. Among the equations we will study include the heat equation, the linear Boltzmann equation, the linear advection equation, and the linear Schr\"odinger equation. The treatment of the first two equations--both are dissipative--will be different than the last two which are wave type equations.  The coefficients that contain uncertainties is assumed to have the following form:
\begin{equation}\label{a-generalform}
a(x,z)=\sum_{i=1}^L a_i(z)b_i(x)
\end{equation}
where $x \in \mathbb{R}^d$ is the space variable, and $z$ is the (possibly high-dimensional) random or deterministic parameters that model uncertainties. Without loss of generality, we assume $a_i(z)>0$ for all $i$  (since otherwise one can absorb the negative sign into $b_i(x)$), and
\begin{equation}\label{a-bound}
\sum_{i=1}^L a^2_i(z)\le C
\end{equation}
for some $C$ independent of $L$.  This already covers a fairly general class of inhomogeneous and uncertain coefficients. This is because one can often approximate a general inhomogeneous coefficient $a(x,z)$ by a linear combination of basis functions in a suitable space, for example in $x$ with an orthogonal basis in $L^2$ (like in the Galerkin approximation in space), or in $z$ through polynomial chaos expansion or Karhunen-Loeve expansion of a random or stochastic process \cite{GS03, ghanem2017, xiu2010}, which gives an approximation in the form of
Eq.~\eqref{a-generalform}. \\

In each of these cases we will estimate the query and gate complexities of the quantum algorithms to compute ensemble averaged physical observables,  where the cost can be \textit{independent} of $M$. For large $M$, advantage with respect to $L$ is possible. In certain regimes, the new quantum algorithms also offer potential speedup with respect to dimension $d$ and precision $\epsilon$. \\

In Section~\ref{sec:quantumintro} we briefly review quantum algorithms for the system of linear equations which are important subroutines for solving ODEs and PDEs. In Section~\ref{sec:PDE} we introduce new transformations for linear PDEs with uncertainty -- shown explicitly for the heat equation, the linear Boltzmann equation, the linear advection equation, and the linear Schr\"odinger equation -- so the equations become deterministic in higher dimensional phase space. For each equation, we show the end-to-end quantum costs (including initial state preparation and final measurement costs) for computing physical quantities of interest in these examples and show when quantum advantage can be expected. Our main results are summarised in Table~\ref{table1}.

    \begin{table}[ht] \label{table1}
    \centering
\caption{Quantum ($\mathcal{Q}$) and classical ($\mathcal{C}$) cost comparison for $r^{\text{th}}$-order approximations in computing ensemble averaged solutions, at $\Lambda$ final meshpoints, over $M$ samples, when $n_0^2 \Lambda=O(N^b)$, where $n_0^2$ is a normalisation of the initial state. Quantum advantages are possible when $\gamma_i>0$ for the corresponding parameters.  We also give the sufficient range of $b$ where quantum advantage is possible. In the table $M_{heat} \equiv O(L^{2+(d+L+3)/c}(d/\epsilon)^{(L+1)/3})$, $M_{Boltz} \equiv O(L^{(2d+L+1)/c}\max(L,d)(d/\epsilon)^{L/c}/d)$, $M_{adv} \equiv O(L^{(d+2L+3)/c+2}(d/\epsilon)^{(2L+2)/c})$ and $M_{Schr} \equiv O((d+L)^{(d+2L+2)/c+2}/(d^{(d+2)/c+2}\epsilon^{2L/r}))$. See text in Section~\ref{sec:PDE} for details.} 
\begin{align}
 \mathcal{O}\left(\frac{\mathcal{C}}{\mathcal{Q}}\right)=\tilde{O}\left(M^{\gamma_1}d^{\gamma_2}L^{\gamma_3} (d+L)^{\gamma_4} \left(\frac{1}{\epsilon}\right)^{\gamma_{5}}\right) \nonumber 
    \end{align} 
    \makebox[1 \textwidth][c]{
    
\begin{tabular}{c c c c c c c c c } 
\hline\hline 
$(d+1)$-dim PDE & \quad $\gamma_1$ &  \quad \quad $\gamma_2$ &  \qquad $\gamma_3$ & $\gamma_4$ & $\gamma_5$ &  $b$ range  &    Parameters \\
($M$ initial data) &  & & & & &  &   with advantage  \\ 
& & & & & &  &   (possible)  \\[2ex]
\hline 
\\

Linear heat  & & & & & &   & & \\
$M<M_{heat}$ &\quad $1$ & \qquad $\frac{d-7-b}{r}-1$ &  \qquad $-4-\frac{9+b}{r}$&  $0$ & $\frac{d-7-b}{r}$ &  $[0, d-7-r]$ &  $M, d, \epsilon$\\
$M>M_{heat}$ &  \quad$0$ & \qquad $\frac{d+L-6-b}{r}-2$ & \qquad  $\frac{d+L-6-b}{3}$ &  $0$& $\frac{d+L-6-b}{r}-1$&  $[0, d+L-6-2r]$ &  $L, d, \epsilon$\\ [3ex]

Linear Boltzmann  & & & & & &   & & \\
$M<M_{Boltz}$ & \quad $1$ & \qquad $\frac{d-2-b}{r}+1$ &  \qquad $-1-\frac{3+d+b}{r}$& $-3$  &  $\frac{d-2-b}{r}-1$ &  $[0, d-2-2r]$ &  $M, d, \epsilon$\\
$M>M_{Boltz}, L>d$ & \quad $0$ & \qquad $\frac{d+L-2-b}{r}$ & \qquad  $\frac{d+L-2-b}{r}$ & $-3$ & $\frac{d+L-2-b}{r}-1$ & $[0, d+L-2-2r]$ &  $L, d, \epsilon$\\
$M>M_{Boltz}, L<d$ & \quad $0$ & \qquad $\frac{d+L-2-b}{r}+1$ & \qquad  $\frac{d+L-2-b}{r}-1$ & $-3$ &  $\frac{d+L-2-b}{r}-1$ & $[0, d+L-2-2r]$ &  $L, d, \epsilon$\\ [3ex]

Linear advection  & & & & & &   & & \\
$M<M_{adv}$ & \quad $1$ & \qquad  $\frac{d-8-b}{r}-2$ &  \qquad $-4-\frac{9}{r}$& $0$  & $\frac{d-8-b}{r}-1$ &  $[0, d-8-2r]$ &  $M, d, \epsilon$\\
$M>M_{adv}$ & \quad$0$ & \qquad $\frac{d+2L-6-b}{r}-2$ & \qquad  $\frac{d+2L-6-b}{r}-2$ & $0$ & $\frac{d+2L-6-b}{r}-1$ & $[0, d+2L-6-2r]$ &  $L, d, \epsilon$\\ [3ex]

Schr\"odinger  & & & & & &   & & \\
$M<M_{Schr}$ & \quad $1$ & \qquad  $\frac{d+2}{r}+2$ &  \qquad $0$&   $-4-\frac{6+b}{r}$ &  $\frac{d-4-b}{r}-1$ & $[0, d-4-r]$  &  $M, d, \epsilon$\\
$M>M_{Schr}$ & \quad $0$ & \qquad $0$ & \qquad  $0$ & $\frac{d+2L-4-b}{r}-2$& \qquad $\frac{d+2L-4-b}{r}-1$ &  \quad $[0, d+2L-4-2r]$ &  $L, d, \epsilon$\\ [3ex]
\hline 
\end{tabular}}

\label{table:summary} 
\end{table}

\section{Quantum algorithms for systems of linear equations} \label{sec:quantumintro}

Numerical methods to solve linear ODEs and PDEs, with solutions at the mesh points and/or time steps given by a vector $x$,  can be written as a system of linear algebraic equations $\mathcal{M}x=y$.    Without loss of generality, we can assume the matrix  $\mathcal {M}$ to be Hermitian (since a general square matrix can be made Hermitian through dilation), with the $(i,j)$-th entry denoted $\mathcal{M}_{ij}$. Assume  $\|\mathcal{M}\|_{max}=\max_{i,j}(|\mathcal{M}_{ij}|)\le 1$ (otherwise one can divide $\mathcal{M}$ by a constant so this condition is satisfied). Its condition number is $\kappa$, and sparsity (the number of non-zero entries in each row and column) is $s$.  Although quantum subroutines do not output all the solutions $x$, many quantum subroutines exist to solve an alternative problem. This is the quantum linear systems problem (QLSP) that outputs the quantum state $|x\rangle$ instead of the full classical solution $x$. 

Assume vectors $x$ and $y$ have elements $\{x_i\}$, $\{y_i\}$.
One can then define the following $m$-qubit quantum states $|x\rangle \equiv \frac{1}{\mathcal{N}_{x}}\sum_i x_i |i\rangle$,  $|y\rangle \equiv \frac{1}{\mathcal{N}_{y}}\sum_i y_i |i\rangle$ where $\mathcal{N}_{x}=\sqrt{\sum_i |x_i|^2}=\|x\|$, $\mathcal{N}_{y}=\sqrt{\sum_i |y_i|^2}=\|y\|$ are normalisation constants.

\begin{problem} \label{prob:one}
\textbf{(QLSP)} Let $\mathcal{M}$ be a $2^m\times 2^m$ Hermitian matrix with spectral norm $\|\mathcal{M}\|\le 1$ with condition number $\kappa$. Assume vectors $x$ and $y$ solve $\mathcal{M}x= y$.  The aim of any QLSP algorithm is, when given access to $\mathcal{M}$ and unitary $U_{initial}$ (where $U_{initial}|0\rangle=|y\rangle$), to prepare the quantum state $|x'\rangle$ that is $\eta$-close to $|x\rangle$, i.e., $\| |x'\rangle-|x\rangle \|  \leq \eta$.
\end{problem}

In Problem~\ref{prob:one}, `access to $\mathcal{M}$' is defined with respect to an oracle that can access the entries of $\mathcal{M}$. There are primarily two types of oracles considered, called sparse-access and block-access respectively, defined below. The query complexity of the algorithm is defined with respect to the number of times such an oracle is used during the protocol. 

\begin{definition}
Sparse access to a Hermitian matrix $\mathcal{M}$ with sparsity $s$ is defined to be the $4$-tuple $(s, \|\mathcal{M}\|_{max}, O_{\mathcal{M}}, O_{F})$.  Here $s$ is the sparsity of $\mathcal{M}$ and $\|\mathcal{M}\|_{max}=\max_{i,j}(|\mathcal{M}_{ij}|)$ is the max-norm of $\mathcal{M}$. $O_M$ and $O_F$ are unitary black boxes which can access the matrix elements $\mathcal{M}_{ij}$ such that 
\begin{align}
    &O_M|j\rangle|k\rangle|z\rangle=|j\rangle|k\rangle|z\oplus\mathcal{M}_{jk}\rangle \nonumber \\
    &O_{F} |j\rangle|l\rangle=|j\rangle|F(j,l)\rangle
\end{align}
where the function $F$ takes the row index $j$ and a number $l=1,2,...,s$ and outputs the column index of the $l^{\text{th}}$ non-zero elements in row $j$. 
\end{definition}
For instance, the common subroutines for QLSP that use sparse-access are the HHL algorithm \cite{HHL-2009} and the linear combination of unitaries (LCU) algorithms \cite{Childs-2021}, including those versions that use the variable-time amplitude amplification algorithm (VTAA) \cite{Ambainis-2012}. The HHL algorithm has query complexity $\tilde{\mathcal{O}}(s\kappa^2\/\eta)$ and the LCU algorithms have the complexity $\tilde{\mathcal{O}}(s \kappa^2 \text{poly}\log(1/\eta))$, where VTAA can be used to reduce each complexity by a factor of $\kappa$. \\

On the other hand, more recent algorithms like those preparing $|x\rangle$ via the quantum adiabatic theorem \cite{Lin2020optimalpolynomial, costa2021optimal}, which demonstrate improved scaling with respect to conditions number $\kappa$ of $\mathcal{M}$ and error $\eta$, often use block-access. The optimal scaling is achieved in \cite{costa2021optimal} with scaling $\mathcal{O}(\kappa \log(1/\eta))$. Since sparse-access to $\mathcal{M}$ can be used to construct block-access to $\mathcal{M}$ \cite{low2019hamiltonian, barrynew}, query complexity with respect to block-access can be converted into sparse-access results with an extra factor of up to $s$ and some small overhead. 
\begin{definition}
Let $\mathcal{M}$ be a $m$-qubit Hermitian matrix, $\delta_{\mathcal{M}}>0$ and $n_{\mathcal{M}}$ is a positive integer. A $(m+n_{\mathcal{M}})$-qubit unitary matrix $U_{\mathcal{M}}$ is a $(\alpha_{\mathcal{M}}, n_{\mathcal{M}}, \delta_{\mathcal{M}})$-block encoding of $\mathcal{M}$ if 
\begin{align} \label{G-def}
    \|\mathcal{M}-\alpha_{\mathcal{M}}\langle 0^{n_{\mathcal{M}}}|U_{\mathcal{M}} |0^{n_{\mathcal{M}}} \rangle \|\leq \delta_{\mathcal{M}}.
\end{align}
Block access to $\mathcal{M}$ is then the 4-tuple $(\alpha_{\mathcal{M}}, n_{\mathcal{M}}, \delta_{\mathcal{M}}, U_{\mathcal{M}})$ where $U_{\mathcal{M}}$ is the unitary black-box block-encoding of $\mathcal{M}$. 
\end{definition}
In the rest of the paper, we assume that if the block access $(\alpha_{\mathcal{M}}, n_{\mathcal{M}}, \delta_{\mathcal{M}}, U_{\mathcal{M}})$ to $\mathcal{M}$ is given, then $U^{\dagger}_{\mathcal{M}}$, controlled-$U_{\mathcal{M}}$ and controlled-$U^{\dagger}_{\mathcal{M}}$ are also given. \\

However, for a quantum algorithm to be useful for PDE problems in real applications, it should output values of meaningful physical quantities of interest, instead of quantum states $|x\rangle$. Thus it is important to solve instead expectations values with respect to $|x\rangle$. 

\begin{problem} \label{prob:two}
Given a Hermitian matrix $\mathcal{G}$, which is of the same size as $\mathcal{M}$, access to $\mathcal{M}$ and $U_{initial}$, the aim is to compute the expectation value $\vect{x}^T \mathcal{G}\vect{x}$ to precision $\epsilon$.
\end{problem}
We note that here the error $\epsilon$ in the expectation value is different to $\eta$, which is the error in the normalised quantum state itself.  

With respect to sparse access, \textbf{Problem 2} can be solved using the quantum singular value decomposition \cite{qsvd} algorithm, which has the following query and gate complexities. 
\begin{lemma} \cite{jin-liu-2022, barrynew} \label{lem:bslep2}
 A quantum algorithm can be constructed that takes the following inputs: (i) sparse access \\
 $(s, \|\mathcal{M}\|_{max}, O_{M}, O_{F})$ to a $2^m \times 2^m$ invertible Hermitian matrix $\mathcal{M}$ such that $\|\mathcal{M}\| \leq s\|\mathcal{M}\|_{max}$; (ii) $m$-qubit unitary $L(t_K, \vect{j}/N)$ where $L(t_K, \vect{j}/N)|0\rangle=|G_{K, \vect{j}}\rangle$; (iii) an accuracy $\epsilon' \in [1/2^m, 1]$ and (iv) $m$-qubit unitary black box $U_{initial}$ where $U_{initial}|0\rangle=|y\rangle$. The algorithm then returns with probability at least $2/3$ an $\epsilon$-additive approximation to $\vect{x}^T \mathcal{G} \vect{x}$ where $\mathcal{G} \equiv |G_{K, \vect{j}}\rangle \langle G_{K,\vect{j}}|$. This algorithm makes $\mathcal{O}(\kappa^2\mathcal{N}_y^2/(\|\mathcal{M}\|\epsilon))$ queries to $U_{\mathcal{G}}$ and $U_{initial}$,  $\tilde{\mathcal{O}}(s\|\mathcal{M}\|_{max}\kappa^3\mathcal{N}_y^2/(\|\mathcal{M}\|\epsilon))$ queries to sparse oracles for $\mathcal{M}$ and $\tilde{\mathcal{O}}(\kappa^2\mathcal{N}_y^2(m+ s\|\mathcal{M}\|_{max}\kappa)/(\|\mathcal{M}\|\epsilon))$ additional $2$-qubit gates. 
\end{lemma}
This scaling can be further improved with respect to $\kappa$, by employing an alternative set of algorithms based on quantum adiabatic computation. Given block-access to $\mathcal{M}$, an optimal scaling in terms of both $\kappa$ and $\epsilon$ can for instance be achieved by combining quantum algorithms for computing observables via block-encoding \cite{rall2020quantum} (optimal scaling in $\epsilon$) and the optimal linear systems solver via the discrete adiabatic theorem \cite{costa2021optimal} (optimal scaling in $\kappa$ and $\epsilon$). 

\begin{lemma} \label{lem:bslep3}
 A quantum algorithm can be constructed that takes the following inputs: (i) $(1, n_{\mathcal{M}}, 0, U_{\mathcal{M}})$ block access to a $2^m \times 2^m$ invertible Hermitian matrix $\mathcal{M}$; (ii) $(1, n_{\mathcal{G}}, 0, U_{\mathcal{G}})$ block access to $2^m \times 2^m$ density matrix $\mathcal{G}$; (iii) an accuracy $\epsilon$; and (iv) $m$-qubit unitary black box $U_{initial}$ where $U_{initial}|0\rangle=|y\rangle$. The algorithm then returns with constant probability an $\epsilon$-additive approximation to $\vect{x}^T \mathcal{G}\vect{x}$ by making $\tilde{\mathcal{O}}(\mathcal{N}_x ^2\kappa /(||\mathcal{M}||\epsilon))$ queries to $U_\mathcal{M}$, $U_{\mathcal{G}}$ and $U_{initial}$. 
\end{lemma}

\begin{proof}
Suppose we are given a $(1, n_{\mathcal{M}}, 0, U_{\mathcal{M}})$ block-encoding of the matrix $\mathcal{M}$. Then results of \cite{Lin2020optimalpolynomial, costa2021optimal} show that with $\tilde{\mathcal{O}}(\kappa/||\mathcal{M}||\log(1/\eta))$ query calls to $U_{\mathcal{M}}$ and $U_{initial}$, one can prepare the state $|\tilde{x}\rangle$ that is $\eta$-close to the state  $|x\rangle \propto \mathcal{M}^{-1}|y\rangle$. One can then adapt the algorithm from \cite{rall2020quantum} (Lemma 5) to compute the expectation $\langle x|\mathcal{G}|x\rangle$ with amplitude estimation. Here one can produce, with probability $1-\delta$, an $\epsilon'$ estimate of $\langle \tilde{x}|\mathcal{G}|\tilde{x}\rangle$ while using $\mathcal{O}(\log(1/\delta)/\epsilon')$ queries to the preparation of $|\tilde{x}\rangle$ with quantum adiabatic computation and calls to $U_\mathcal{G}$. Combining these results, we obtain a $\epsilon'$-additive approximation to $\langle x|\mathcal{G}|x\rangle$ with $\tilde{\mathcal{O}}(\kappa/\|\mathcal{M}\|\epsilon')$ queries to $U_{\mathcal{M}}$, $U_{\mathcal{G}}$ and $U_{initial}$, where one can choose any constant $\delta>1/2$. Since $\vect{x}^T \mathcal{G}\vect{x}=\mathcal{N}_x^2 \langle x|\mathcal{G}|x\rangle$, then $\epsilon'=\mathcal{N}_x^2 \epsilon$ and we have our result. Since $\mathcal{N}_x$ may not be assumed to be known in advance, unlike $\mathcal{N}_y$, another quantum algorithm with cost $\tilde{\mathcal{O}}(\kappa \mathcal{N}_x/\epsilon'')$ \cite{chakraborty2019power, linden2020quantum} is required to compute $\mathcal{N}_x$ to precision $\epsilon''$. 
\end{proof}

\begin{lemma} \cite{rall2020quantum}
Given a density matrix $\mathcal{G}$ and let $|\rho\rangle=U_{\rho}|0\rangle$ where $U_{\rho}$ consists of $R$ elementary gates. Then for every $\epsilon$, $\delta>0$, there exists a quantum algorithm that estimates $\langle \rho|\mathcal{G}|\rho\rangle$ to precision $\epsilon$ with probability at least $1-\delta$. This algorithm has gate complexity $O((R/\epsilon)\log(1/\delta))$. 
\end{lemma}

Unlike the protocol in Lemma~\ref{lem:bslep2}, the protocol in Lemma~\ref{lem:bslep3} requires the normalisation of the solution state $\mathcal{N}_x$ instead of the initial state normalisation $\mathcal{N}_y$. Since the classical values of the initial state is known, we can assume $\mathcal{N}_y$ to be given. We can bound $\mathcal{N}_x$ using $\mathcal{N}_x=\|\vect{x}\|=\|\mathcal{M}^{-1}\vect{y}\|\leq \|\mathcal{M}^{-1}\|\|\vect{y}\|=\kappa \mathcal{N}_y/\|\mathcal{M}\|$ since $\|\mathcal{M}\|\|\mathcal{M}^{-1}\|=\kappa$. This means that by taking the worst-case scenario Lemma~\ref{lem:bslep3} still requires a cost that scales like Lemma~\ref{lem:bslep2} with respect to $\kappa$ and $\epsilon$. Thus, to give the most conservative costs in this paper, and also to avoid implementing an extra quantum algorithm to compute $\mathcal{N}_x$, it is sufficient for us to use the protocol in Lemma~\ref{lem:bslep2}. 


We note that the protocols above all assume the preparation of the initial state $|y\rangle$. The total gate  complexity in an end-to-end quantum algorithm must also include the gate complexity in the preparation of the initial state. The gate complexity for preparing $|y\rangle$ can be greatly reduced when the state is sparse. So if $|y\rangle$ has sparsity $\sigma$ (i.e. vector $y$ has $\sigma$ non-zero entries), its deterministic preparation can have the following gate complexity 

\begin{lemma} \label{lem:initialprep}
 \cite{gleinig2021efficient} A circuit producing an $m$-qubit state $|y\rangle$ from $|0\rangle$ with given classical entries $\{y_i\}$ can be implemented using $\mathcal{O}(m\sigma)$ CNOT gates and $\mathcal{O}(\sigma(\log \sigma +m))$ one-qubit gates, where the specification of the circuit can be found with a classical algorithm with run-time $\mathcal{O}(m\sigma^2\log \sigma)$. 
\end{lemma}
We note that the above protocol uses a constant number of ancilla qubits. Alternative preparation strategies are proposed in  \cite{zhang2022quantum}, which has a reduced gate complexity $\Theta(\log (m \sigma))$, but require $\mathcal{O}(m \sigma \log \sigma)$ ancilla qubits.

\section{Linear PDEs with uncertainty} \label{sec:PDE}
In this section, we demonstrate our new algorithm for the heat equation, the linear Boltzmann equation, the linear advection equation and the linear Schr\"odinger equation, all with uncertain coefficients of the form in Eq.~\eqref{a-generalform}. 

\subsection{Linear heat equation}

Consider the following initial value problem of the linear heat equation
 \begin{equation}\label{heat}
\begin{cases} \partial_t u - a(x,z) \Delta u=0\, \\
u(0,x, z)=u_0(x,z)
\end{cases}
\end{equation}
where $u=u(t,x, z)$, $x \in \mathbb{R}^d$ is the position, $t\ge 0$ is the time, $a(x,z)> 0$ is heat conductivity given by Eq.~\eqref{a-generalform}, due to the heterogeneity of the media or background (for example in porous media \cite{zhang2001stochastic}).   To perform a Monte-Carlo simulation one needs to select a large number (say $M \gg 1$) of samples in $z$, solve the PDE system $M$ times and then take the ensemble average. However, solving the equation $M$ times for large $M$ is computationally expensive. The computation is also costly for large $L$.

We aim instead to find a transformation to another linear PDE with \textit{certain} coefficients and only a \textit{single} initial condition, that still enables us to compute ensemble averaged observables corresponding to the original PDE.

 Let $ p=(p_1, \cdots, p_L)^T$, with $p_{i}\in (-\infty,\infty)$ for all $i=1,...,L$.   We introduce the transformation
\begin{align}\label{U-def-11}
U(t,x,z,p)=\frac{1}{2}\prod_{i=1}^La_i(z)e^{-a_i(z)  |p_i|}u(t,x,z),
\end{align}
from which one can recover $u$ from $U$ via
\begin{align}\label{U-to-u-11}
u(t,x,z)=\int_{(-\infty,\infty)^L} U(t,x,z,p) dp=2\int_{(0,\infty)^L} U(t,x,z,p) dp.
\end{align}
A simple computation shows that $U$ solves
 \begin{equation}\label{phase-heat-11}
\partial_t U + \sum_{i=1}^L \text{sign} (p_i)  b_i(x) \Delta \partial_{p_i} U = 0,
\end{equation}
in which the coefficients of the equation are \textit{independent} of $z$! 

To make sure the solution to new Eq.~\eqref{phase-heat-11} is mathematically well-defined,  we first assume $b_i$ is independent of $x$ for simplicity and take a Fourier Transform on $x$
\begin{equation}\label{heat-Fourier}
     \partial_t\hat{U} -|\xi|^2 \sum_{i=1}^L {\text {sign}}(p_i)  b_i \partial_{p_i}{\hat{U}}=0,
 \end{equation}
  where the Fourier variable in $x$ is denoted by $\xi\in \mathbb{R}^d$. 
This is a linear transport equation for $\hat{U}$ with a {\it discontinuous} coefficient. Often one needs to make sense of the solution at $p=0$, depending on the physical background of the problem. One possibility is to provide a physically relevant
jump condition at $p=0$ \cite{Jin-Hyp}. However, from the definition in Eq.~\eqref{U-def-11}, $U$ is {\it continuous} at $p_i=0$, so there is no need to  impose any interface condition for $U$ at $p_i=0$. This can also be justified from the solution to Eq.~\eqref{heat-Fourier}. For clarity, consider the case of $L=1$ and $b=1$. Then Eq.~\eqref{heat-Fourier} is
\begin{equation}\label{heat-Fourier-1}
     \partial_t\hat{U} -\xi^2 {\text {sign}}(p)   \partial_{p}{\hat{U}}=0.
 \end{equation}
 By the method of characteristics, 
 \begin{equation}
   \hat{U}(t,\xi,z, 0^-)=\hat{U} (0,\xi, z,-\xi^2 t), \qquad
  \hat{U}(t,\xi,z,0^+)=\hat{U} (0,\xi,z, \xi^2 t).
 \end{equation}
From the definition in Eq.~\eqref{U-def-11} one clearly sees that $\hat{U} (0,\xi, z,-\xi^2 t)=\hat{U}(0,\xi,z, \xi^2 t)$, therefore 
\begin{equation}
    \hat{U}(t,\xi,z, 0^-)=\hat{U}(t,\xi,z, 0^+)
\end{equation}
and $\hat{U}$ is continuous at $p=0$.

By method of characteristics, it is also easy to check that $\hat{U} $ is an even function for each $p_i$, since its initial data is also an even function.
\\

To work with $M$ samples $\{z_m\}$, $m=1,...,M$, we now define
\[
V(t, x, p)=\frac{1}{M}\sum_{m=1}^M U(t,x,z_m,p)
\]
for $p_i\in (-\infty, \infty), \, i=1,\cdots, L$, which solves
 \begin{align}\label{phase-heat3-11}
\begin{cases} \partial_t V + \sum_{i=1}^L {\text{sign}} (p_i) b_i(x)\Delta \partial_{p_i} V =0\, \\
V(0,x,p)= \frac{1}{M}\sum_{m=1}^M \prod_{i=1}^L a_i(z_m) e^{- a_i(z_m)|p_i|} u(0,x,z_m)\,.
\end{cases}
\end{align}
This is the  linear PDE we will solve, which has \textit{certain} coefficients and a \textit{single} initial condition. Now the average of the solutions of the original problem, Eq.~\eqref{heat}, can be recovered from $V(t,x,p)$ using
\begin{align}\label{u-average-11}
\overline{u}(t,x)&=\frac{1}{M}\sum_{m=1}^M u(t,x, z_m)
= \int V(t, x, p)\, dp \,.
 \end{align}
Thus, in solving for $V$, the computational cost is clearly {\it independent of} $M$!\\

To solve Eq.~\eqref{phase-heat3-11} numerically to, say $t=1$, we
use $N_t$ steps in time: $t=n\Delta t$ for $k=0,...,N_t-1$ where $N_t\Delta t=1$. Throughout the paper we assume the spatial computational domain to be $[0,1]^d$, with $N$ spatial mesh points for each dimension. For example, in one space dimension $x=j\Delta x$ for $j=0,...,N-1$ where $N\Delta x=1$. Since $V$ decays exponentially in $|p|$, one can truncate the computational domain in finite $p$, say at $p_i=\pm 1$, namely $p_i\in [-1,1]$ for all $1\le i\le L$.  Define the discrete $p$ as $p=-1+k \Delta p$ for $k=0,...,N_p-1$, where $N_p \Delta p=2$. In $d$ spatial dimensions, the total number of spatial meshpoints for $x$ is $N^d$ and for $p$ is $N_p^L$ since $p$ is an $L$-dimensional vector. 

We use the center finite difference in space, the upwind scheme in $p$ (so the overall spatial error is of $O(1/N)$) and the forward Euler scheme in time (so the time error is of $O(\Delta t)$) to solve the phase space heat equation in Eq.~\eqref{phase-heat3-11}. See Appendix \ref{AppendixB-1} for details of the discretisation scheme and the corresponding properties of matrix $\mathcal{M}$.
 
 We  use the quadrature rule for the integration in \label{u-average} using  the discrete values $V_{N_t,k}=V(t=N_t\Delta t, p=k\Delta p)$. If the numerical solution vector is denoted $\vect{V}=\sum_n^{N_t}\sum_j^{N^d}\sum_k^{N_p^L} V_{n,j,k}|n \rangle |j\rangle |k\rangle$, then the ensemble average of $u$ at point $(t=N_t\Delta t, x=J \Delta x)$ is 
\begin{align}
   \bar{u}(t,x)=\int V(t,x,p)dp \approx \bar{u}_{N_t,J} \equiv \frac{1}{N_p^L}\sum_k^{N_p^L}V_{N_t, J, k}
\end{align}
and $\vect{V}^T \mathcal{G}\vect{V}=N_p^L |\bar{u}_{N_t,J}|^2$, where $\mathcal{G}=|G_{N_t,J}\rangle \langle G_{N_t,J}|$ with $|G_{N_t,J}\rangle=(1/\sqrt{N_p^L})\sum_{k}^{N_p^L}|N_t\rangle |J\rangle |k\rangle$. The quantum state embedding of the initial condition of the state is defined as 
\begin{align}
    |V_0\rangle=\frac{1}{\mathcal{N}_{V_0}}\sum_{j}^{N^d}\sum_k^{N_p^L}V_{0,j,k} |0\rangle |j\rangle |k \rangle,
\end{align}
where $\mathcal{N}^2_{V_0}=\sum_{j}^{N^d}\sum_k^{N_p^L}|V_{0,j,k}|^2$ and the state is assumed to have sparsity $\sigma_{V_0}$.  

\begin{lemma} \label{lem:V0constant}
The constant $n^2_{V_0} \equiv \mathcal{N}^2_{V_0}/N_p^{L}$ has range $O(1) \leq n^2_{V_0} \leq O(N^{d})$. Different $n_{\psi_0}$ corresponds to different initial data. If we assume the initial data has support in a box of size $\beta \in (0,1)$, then $n_{V_0}^2=O((\beta N)^d)$. 
\end{lemma}

\begin{proof}
We first prove that $\int dp\, |V(0, x, p)|^2\le C=O(1)$. First, it is easy to check that 
$\int dp\, |U(0, x, p,z)|^2=\frac{1}{4}\prod_{i=1}^L a_i(z)^2 |u_0(x,z)|^2\le C=O(1)$. Here the boundedness of $\prod_{i=1}^L a_i(z)^2$ is a result of the boundedness of $\sum_{i=1}^L a_i(z)^2$ (from Eq.~\eqref{a-bound}). Then 
\[
\int dp\, |V(0, x, p)|^2=\int dp \left(\frac{1}{M}\sum_{m=1}^M  U(0, x, p,z_m)\right)^2\le \frac{1}{M}\sum_{m=1}^M  
\int dp \,U(0, x, p, z_m)^2\le C=O(1).
\]
Then using the quadrature rule we have 
$O(1)\ge \int dx \int dp\, |V(0, x, p)|^2 \approx (1/(\beta N)^d)\sum_j^{N^d} \int dp\, |V(0, x=j\Delta x, p)|^2 $. Then
$O(1) \geq \int dx \int dp \,|V(0,x, p)|^2 \approx \mathcal{N}^2_{V_0}/((\beta N)^d N_p^L)$. Thus $O(1) \leq n^2_{V_0} \equiv \mathcal{N}_{V_0}^2/N_p^L \leq O((\beta N)^d)$. 
\end{proof}
Then the quantum algorithm needed to recover $\bar{u}$  has the following complexity.

 \begin{theorem} \label{thm:linearheat} 
     A quantum algorithm that takes sparse access to $\mathcal{M}$ (using an $r^{\text{th}}$-order method with $r\geq 1$) and access to classical values of the initial conditions $V_{0, J, l}$, is able to estimate the density $|\bar{u}_{N_t, J}|^2$ to precision $\epsilon$ at $\Lambda$ meshpoints, with an upper bound on the gate complexity
     \begin{align}
         \mathcal{Q}=\tilde{\mathcal{O}}(n^2_{V_0}\Lambda (Ld)^{3}(Ld/\epsilon)^{1+9/r})
     \end{align}
     and a smaller query complexity. 
 \end{theorem}
 
 \begin{proof}
     The sparsity and condition number of the corresponding matrix $\mathcal{M}$ for Eq.~\eqref{phase-heat3-11} is $s=O(Ld)$ and $\kappa=O(LdN^3)=O(Ld(Ld/\epsilon)^{3/r})$ respectively, where $N=O((Ld/\epsilon)^{1/r})$ from Appendix~\ref{app:classicalheat}. From Lemma~\ref{lem:bslep2} the cost in recovering $\vect{V}^T \mathcal{G} \vect{V}=N_p^L \bar{u}_{N_t, J}^2$ to precision $\epsilon'$ is $\tilde{\mathcal{O}}(s \kappa^3 \mathcal{N}_{V_0}^2/\epsilon')=\tilde{\mathcal{O}}(s \kappa^3 n_{V_0}^2/\epsilon)=\tilde{\mathcal{O}}(n_{V_0}^2(Ld)^3(Ld/\epsilon)^{1+9/r})$ where $\epsilon'=N_p^L \epsilon$, $\epsilon$ is the error in $\bar{u}_{N_t,J}^2$. Since $\Lambda$ different states $|G_{N_t,J}\rangle$ are required for each mesh point $(N_t,J)$, so the total cost must be multiplied by $\Lambda$.
 \end{proof}
 
 \begin{lemma} \label{lem:classicallinearheatextension}
 When $M<M_{heat} \equiv O(L^{2+(d+L+3)/r}(d/\epsilon)^{(L+1)/3})$, the classical algorithm has minimal cost $\mathcal{C}=O(Md^2(d/\epsilon)^{(d+2)/r})$. If $M>M_{heat}$, the classical algorithm has minimal cost $\mathcal{C}=O(L^{2+(d+L+3)/r}d^2(d/\epsilon)^{(d+L+3)/r})$. 
 \end{lemma}
 \begin{proof}
See Appendix~\ref{app:classicalheat} for details. 
 \end{proof}

Now we define what we mean by quantum advantage with respect to classical algorithms with cost $\mathcal{C}$. 

\begin{definition}
 We say there is a quantum advantage in estimating the quantity of interest when $\mathcal{Q}=\littleo(\mathcal{C})$.
\end{definition}

 \begin{corollary} \label{cor:linearheatextension}
   To attain a quantum advantage when $M<M_{heat}$, it is  sufficient for the following condition to hold
\begin{align} \label{eq:heatadvantage1}
\littleo\left(\frac{M}{n^2_{V_0}\Lambda L^{4+9/r}d}\left(\frac{d}{\epsilon}\right)^{(d-7)/r}\right)=\tilde{O}(1). 
\end{align}
When $M>M_{heat}$, it is sufficient that the following is satisfied
\begin{align}
  \littleo\left(\frac{1}{n^2_{V_0}\Lambda Ld}\left(\frac{Ld}{\epsilon}\right)^{(d+L-6)/r-1}\right)=\tilde{O}(1).   
\end{align}
 \end{corollary}
 \begin{proof} Combining Lemma~\ref{lem:classicallinearheatextension} and Theorem~\ref{thm:linearheat}.
 \end{proof}
 
We see from Eq.~\eqref{eq:heatadvantage1} that quantum advantage with respect to $M$ is always possible in the range $M'_{heat}<M<M_{heat}$ where $M'_{heat}=\tilde{\mathcal{O}}(n^2_{V_0}\Lambda L^{4+9/r}d(\epsilon/d)^{(d-7)/r})$, with up to exponential quantum advantage in $d$ and $\epsilon$. As a simple example, if we begin with a point source initial condition and only require the final solution at $O(1)$ points, then $n^2_{V_0}=O(1)=\Lambda$. When $d>7$, the range of possible $M$, captured by $M_{heat}-M'_{heat}$, is very large. A necessary condition for such an $M$ to exist in more general cases is $O(1)<n^2_{V_0}\Lambda< \tilde{\mathcal{O}}(L^{(d+L-6)/r-2}(d/\epsilon)^{(L+1)/3-(d-7)/r}/d)$, which allows for a wide range of possibilities. 

In the case when $M>M_{heat}$, we see there is no quantum advantage with respect to $M$, but there is quantum advantage in $L, d, \epsilon$ when $n^2_{V_0}\Lambda<\tilde{\mathcal{O}}((Ld/\epsilon)^{(d+L-6)/r-1}/(Ld))$.

  \subsection{The linear Boltzmann equation}
  
 Now consider the linear Boltzmann equation with isotropic scattering \cite{lewis1984computational}
 \begin{equation}\label{Transp}
\begin{cases} \partial_t f + v\cdot \nabla_x f = a(x, z) \left[
\frac{1}{\Omega}\int_{\mathbb{R}^d} f\, dv - f\right]\,, \\
f(0,x, z)=f_0(x,z)\,,
\end{cases}
\end{equation}
with suitable boundary conditions. Here $f(t,x, v, z)>0$ is the particle density distribution, $x \in \mathbb{R}^d$ is the position,
$v \in \mathbb{R}^d$ is the velocity of the particles, $a(x, z)> 0$, defined in \eqref{a-generalform}, is the scattering rate and $\Omega$ is the volume of the domain for $v$. Here $a(x,z)$ is often uncertain due to experimental or modeling errors \cite{fichtl2009stochastic}. The left hand side of the equation models particle transport, while the  right hand side models scattering collision of particles with the background.  

We introduce the following transformation
\begin{equation} \label{Big-F-Transp-10}
  F'(t, x,v,z, p) = \frac{1}{2} \prod_{i=1}^L a_i(z) e^{-|p_i| a_i(z)} f(t,x,v, z)
 \end{equation}
where  $p_i\in(-\infty, \infty), i=1, \cdots, L$. Then $f$ can be recovered from $F$ via
 \[
 f(t,x,v,z)=\int_{(-\infty,\infty)^L} F'(t,x,v,z,p)=2\int_{[0,\infty)^L} F'(t,x,v,z,p) \, dp,
 \]
 and 
 $F'$ solves
 \begin{equation}\label{Boltzmann-phase-11}
    \partial_t F' + v\cdot \nabla_x F' = -\sum_{i=1}^L  {\text{sign}} (p_i)\left[
\frac{1}{\Omega}\int_{\mathbb{R}^d}  b_i(x)\partial_{p_i} F'\, dv -  b_i(x) \partial_{p_i} F'\right]\,.
\end{equation}

Then the following stability and conservation results hold. 

\begin{theorem} \label{Thm-F}
Assume vanishing boundary condition for $F'$ and $
\sum_{i=1}^L b_i(x)\ge 0$. Then
\begin{eqnarray}
&\partial _t \int_{\mathbb{R}^d}\int_{\mathbb{R}^d} F' \, dv\, dx =0 \label{F-1}\\
&\partial_t \int_{\mathbb{R}^d}\int_{\mathbb{R}^d}\int_{[0,\infty)^L} (F')^2 \, dp\, dv\,  dx \le 0\,. \label{F-2}
\end{eqnarray}
\end{theorem}

\begin{proof}
Eq. (\ref{F-1})  is easy to see by integrating over $v$ and $x$ in \eqref{Boltzmann-phase-11} and using the vanishing boundary condition. Since $F'$ is an even function in $p_i$, we will prove \eqref{F-2} for the case of $p_i>0$ for all $i$.  Multiplying 
\eqref{Boltzmann-phase-11} by 2F' gives
\begin{equation}
    \partial_t (F')^2 +  v\cdot \nabla_x (F')^2 = -\sum_{i=1}^L b_i(x) \left[
\frac{2}{\Omega}F' \int_{\mathbb{R}^d} \partial_{p_i} F'\, dv - \partial_{p_i} (F')^2\right]\,.
\end{equation}
Integrating over $v$, one gets
\begin{equation}
    \partial_t \int_{\mathbb{R}^d} (F')^2 dv +    \nabla_x \cdot \int_{\mathbb{R}^d} v (F')^2 dv = - \sum_{i=1}^L b_i(x) 
   \left[ \frac{1}{\Omega}\partial_{p_i} \left(\int_{\mathbb{R}^d} F' dv\right)^2
   - \partial_{p_i} \int_{\mathbb{R}^d} (F')^2\, dv \right] \,.
\end{equation}
Now integrating over $p$. Denoting $p_i^-=(p_1, \cdots, p_{i-1},p_{i+1}, \cdots, p_L)$ and using $F'(t,x,v,p)|_{p_i=0}$, gives
\begin{eqnarray}
\nonumber && \partial_t\int_{[0,\infty)^L} \int_{\mathbb{R}^d} (F')^2\, dv\, dp +    \nabla_x \cdot \int_{[0,\infty)^L} \int_{\mathbb{R}^d} v (F')^2\, dv\, dp\\
 = && \sum_{i=1}^L b_i(x) \int_{[0,\infty)^{L-1}}
   \left[ \frac{1}{\Omega} \left(\int_{\mathbb{R}^d} F'(t,x,v,z,p)|_{p_i=0}\, dv\right)^2
   -  \int_{\mathbb{R}^d} F'(t,x,v,z,p)^2|_{p_i=0}\, dv \right] dp_i^- \le 0 \,,
\end{eqnarray}
where the inequality is obtained by Jensen's inequality. 
Now integrating over $x$, \eqref{F-2} is proved.
\end{proof}

Suppose one is  interested in computing the ensemble average of $M$ different initial data resulting from $M$ samples of $z$. Then similarly to the heat equation case, we can define
 \begin{align} \label{Big-G-Transp-B-11}
  F(t, x,v, p) = \frac{1}{M}\sum_{m=1}^MF'(t,x,v,z_m,p)=\frac{1}{M} \sum_{m=1}^M\prod_{i=1}^L a_i(z_m) e^{-|p_i| a_i(z_m)} f(t,x,v, z_m)
 \end{align}
to obtain the following problem:
\begin{equation}\label{phase-Transp-11}
\begin{cases}  \partial_t F + v\cdot \nabla_x F = -\sum_{i=1}^L {\text{sign}} (p_i) \left[
\frac{1}{\Omega}\int_{\mathbb{R}^d} b_i(x)\partial_{p_i} F\, dv -  b_i(x) \partial_{p_i} F\right]\,  \\
  F(0,x,v,  p)=\frac{1}{M}\sum_{m=1}^M \Pi_{i=1}^L a_i(z_m)e^{-|p_i|a_i(z_m)}f_0(x,v,z_m)  \,.
\end{cases}
\end{equation}
Then the average of the solutions of the original problem, Eq.~\eqref{Transp}, can be recovered from $F(t,x,p)$ using
\begin{align}\label{f-average-11}
\overline{f}(t,x,v)&=\frac{1}{M}\sum_{m=1}^M f(t,x, v,z_m)
= \int F(t, x,v, p)\, dp \,.
 \end{align}

In practical applications one is interested in the moments of $f$, which give rises to physical quantities of interest including, for example, density $\rho(t,x)=\int f(t,x,v)\,dv = \int\int F(t,x,v,p) \, dp \, dv$, flux $J(t,x)=\int vf(t,x,v)\,dv = \int\int vF(t,x,v, p) \, dp \, dv$, and kinetic energy $E(t,x)=\int \frac{v^2}{2} f(t,x,v)\,dv = \int\int \frac{v^2}{2} F(t,x,v, p) \, dp \, dv$. The ensemble averages of these physical quantities can then be computed using Eq.~\eqref{moments}. \\

Again, notice that $F$ decays exponentially in $p$, which means we can truncate the computational domain in some finite point, without loss of generality, for instance at 
at $p_i=\pm 1$.  We
consider $N_t$ steps in time $t_n=n\Delta t$ for $n=0,...,N_t-1$ where $N_t\Delta t=1$.  We also use $N$ mesh points for each space dimension. For example, in one dimension $x_j=j\Delta x$ for $j=0,...,N-1$ where $N\Delta x=1$, and $p_k=-1+k \Delta p$ for $k=0,...,N_p-1$, where $N_p \Delta p=2$. To simplify the notation we let $N=N_p$.  The discrete-ordinate method (namely, quadrature rules for the collision term) is used \cite{lewis1984computational})  with quadrature points $v_l=l\Delta v$, for  $n=0,...,N_v-1$ where $N_v \Delta v=1$. In $d$ spatial dimensions, the total number of spatial mesh points for $x$ is $N^d$, for $v$ is $N_v^d$ and for $p$ is $N_p^L$ since $p$ is $L$-dimensional.  To solve Eq.~\eqref{phase-Transp-11},  we use the upwind scheme for $x$ and $w$, and forward Euler method in time.  For details of the numerical scheme and the corresponding matrix $\mathcal{M}$ for the problem, see Appendix \ref{AppendixB-2}.  The overall spatial error is of $O(1/N)$ and temporal error is of $O(\Delta t)$.

Let $F_{n, j, l, k}=F(t=n\Delta t, x=j \Delta x, v=l \Delta v, p=k\Delta p)$, $\rho_{n, j}=\rho(t=n\Delta t, x=j \Delta x)$ and $J_{n,j}=J(t=n\Delta t, x=j\Delta x)$, $v_l=l\Delta v$, where $n=1,...,N_t$, $j=1,...,N^d$, $l=1,...,N_v^d$, $k=1,...,N_p^L$. Then one can numerically approximate the density, flux and energy through the following quadrature rule:
\begin{align}\label{moments}
    & \rho(t,x)=\int \int F(t, x, v, p) dv dp \approx \rho_{n, j} \equiv  \frac{1}{N_v^dN_p^L}\sum_{l}^{N_v^d}\sum_{k}^{N_p^L} F_{n,j,l,k}, \\
    & J(t,x)=\int \int vF(t,x,v,p) dv dp \approx J_{n, j} \equiv \frac{1}{N_v^dN_p^L}\sum_{l}^{N^d_v}\sum_{k}^{N_p^L}v_l F_{n, j, l, k}, \\
    & E(t,x)= \int \int (v^2/2) F(t,x,v,p) dv dp \approx E_{n,j} \equiv \frac{1}{2N_v^d N_p^L}\sum_l^{N_v^d}\sum_k^{N_p^L}v^2_lF_{n,j,l,k}.
\end{align}
Denoting the solution vector to Eq.~\eqref{phase-Transp-11} by $\vect{F}=\sum_n^{N_t}\sum_j^{N^d}\sum_l^{N_v^d}\sum_k^{N_p^L}F_{n,j,l,k}|n\rangle|j\rangle|l\rangle|k\rangle$, one can write 
\begin{align}
    &\vect{F}^T \mathcal{G}^{density} \vect{F}=N_v^d N_p^L|\rho_{n, j}|^2, \\
    & |G^{density}_{n,j}\rangle=\frac{1}{\sqrt{N_v^dN_p^L}}\sum_{l}^{N_v^d} \sum_{k}^{N_p^L}|n\rangle |j\rangle |l\rangle |k\rangle, \nonumber \\
    &  \vect{F}^T \mathcal{G}^{flux} \vect{F}=N^L_p\left(\frac{N^d_v}{\mathcal{N}_V}\right)^2|J_{n, j}|^2=O(N_v^dN_p^L|J_{n,j}|^2), \\
     & |G^{flux}_{n,j} \rangle=\frac{1}{\mathcal{N}_V\sqrt{N_p^L}}\sum_{l}^{N_v^d}\sum_{k}^{N_p^L} v_{l} |n\rangle |j\rangle|l\rangle|k\rangle, \nonumber \\
      &  \vect{F}^T \mathcal{G}^{energy} \vect{F}=4N^L_p\left(\frac{N^d_v}{\mathcal{N}_{V^2}}\right)^2|E_{n, j}|^2=O(N_v^dN_p^L|E_{n,j}|^2), \\
     & |G^{energy}_{n,j} \rangle=\frac{1}{\mathcal{N}_{V^2}\sqrt{N_p^L}}\sum_{l}^{N_v^d}\sum_{k}^{N_p^L} v^2_{l} |n\rangle |j\rangle|l\rangle|k\rangle, \nonumber 
\end{align}
where $\mathcal{N}_V^2=\sum_l^{N_v^d}|v_l|^2=\sum_l^{N_v^d} (l\Delta v)^2=O(N_v^{d})$ and $\mathcal{N}^2_{V^2}=\sum_l^{N_v^d}|v^2_l|^2=\sum_l^{N_v^d} (l\Delta v)^4=O(N_v^{d})$. The initial condition can be encoded in $|F_0\rangle$ defined by
\begin{align}
    |F_0\rangle=\frac{1}{\mathcal{N}_{F_0}}\sum_{j}^{N^d}\sum_{l}^{N_v^d}\sum_{k}^{N_p^L} F_{0,j,l,k} |0\rangle |j\rangle|l\rangle|k\rangle,
\end{align}
where the normalisation $\mathcal{N}^2_{F_0}=\sum_{j}^{N^d}\sum_{l}^{N_v^d}\sum_{k}^{N_p^L} |F_{0,j,l,k}|^2$ and the initial state is assumed to  have sparsity $\sigma_{F_0}$. The kinetic energy can be computed similarly and we omit its details.

\begin{lemma} \label{lem:F0constant}
The constant $n^2_{F_0} \equiv \mathcal{N}^2_{F_0}/(N_p^LN_v^{d})$ has range $O(1) \leq n^2_{F_0} \leq O(N^{d})$. Different $n_{\psi_0}$ corresponds to different initial data. If we assume the initial data has support in a box of size $\beta$, then $n^2_{F_0}=O((\beta N)^d)$. 
\end{lemma}

\begin{proof}
First, similar to the proof in Lemma \ref{lem:V0constant},   $\int dv \int dp |F(0, x, v, p)|^2\le O(1)$. Then using the quadrature rule we have 
$\int dx \int dv \int dp |F(0, x, v, p)|^2 \approx (1/(\beta N)^d)\sum_j^{N^d} \int dv \int dp |F(0, x=j\Delta x, v, p)|^2 \leq O(1)$. Then 
$O(1) \geq \int dx \int dv \int dp |F(0,x,v,p)|^2 \approx \mathcal{N}^2_{F_0}/((\beta N)^d N_v^d N_p^L)$. Thus $n^2_{F_0} \equiv \mathcal{N}_{F_0}^2/(N_p^LN_v^d) \leq O((\beta N)^d)$. However, the quantum inner product $\Upsilon\equiv \langle F_0|(\mathcal{M}^{-1})^{\dagger}\mathcal{G} \mathcal{M}^{-1}|F_0\rangle$ is upper bounded by $1$, where $1 \geq \Upsilon=\vect{F}^T \mathcal{G} \vect{F}/\mathcal{N}^2_{F_0}=O(N_v^d N_p^L/\mathcal{N}^2_{F_0})=O(1/n^2_{F_0})$, since the observables $|\rho|^2, |J|^2=O(1)$ for $\mathcal{G}=\mathcal{G}^{density}, \mathcal{G}^{flux}$. 
\end{proof}
\noindent \textbf{Remark:} The definition $n_{F_0} \equiv \mathcal{N}_{F_0}/(N_p^LN_v^d)^{1/2}$ can be an underestimate in situations where the $N_p^LN_v^d$ term in the quadrature sum overestimates the number of non-zero factors in the summation. \\

With this, we can compute the total cost for an `end-to-end' quantum algorithm for computing both the density and flux.
\begin{theorem}
 A quantum algorithm that takes sparse access to $\mathcal{M}$ (using an $r^{\text{th}}$-order method with $r\geq 1$) is able to estimate $|\rho_{N_t, J}|^2$, $|J_{N_t, J}|^2$ and $|E_{N_t, J}|^2$ at $\Lambda$ mesh  points to precision $\epsilon$ with an upper bound on the gate complexity  
\begin{align} \label{eq:qquerytotalboltzmann}
\mathcal{Q}=\tilde{\mathcal{O}}\left(n^2_{F_0}\Lambda (L/\epsilon)(d+L)^3\left(\frac{Ld}{\epsilon}\right)^{(3+d)/r}\right)
\end{align}
with a smaller query complexity.
\end{theorem}

\begin{proof}
The matrix $\mathcal{M}$ corresponding to the discretisation of Eq.~\eqref{phase-Transp-11} has sparsity $s=O(LN^d_v)=O(L(Ld/\epsilon)^{d/r})$ and condition number $\kappa=O((d+L)N)=O((d+L)(Ld/\epsilon)^{1/c})$, where $N=N_p=N_v=O((Ld/\epsilon)^{1/r})$ from Appendix~\ref{app:classicaltransp}. To compute $\vect{F}^T \mathcal{G}^{density} \vect{F}$, $\vect{F}^T \mathcal{G}^{flux} \vect{F}$ and $\vect{F}^T \mathcal{G}^{energy} \vect{F}$ to precision $\epsilon'$, from Lemma~\ref{lem:bslep2} the maximum query complexity is $\tilde{\mathcal{O}}(s\kappa^3\mathcal{N}^2_{F_0}/\epsilon')=\tilde{\mathcal{O}}(s\kappa^3n^2_{F_0}/\epsilon)=\tilde{\mathcal{O}}(n^2_{F_0}(d+L)^3(L/\epsilon)(Ld/\epsilon)^{(3+d)/r})$ where $\epsilon=\epsilon'/(N_p^LN_v^d)$ is the error in $|\rho_{N_t,J}|^2$, $|J_{N_t,J}|^2$ and $|E_{N_t,J}|^2$. The same order of additional 2-qubit gates are required. If one wants the density at $\Lambda$ meshpoints, then total gate complexity must be multiplied by $\Lambda$. The same argument applies to computing the flux.
\end{proof}

 \begin{lemma}
 When $M<M_{Boltz} \equiv O(L^{(2d+L+1)/r}\max(L,d)(d/\epsilon)^{L/r}/d)$, the classical algorithm for the problem has the cost $\mathcal{C}=O(Md (d/\epsilon)^{(2d+1)/r})$. When $M>M_{Boltz}$, the classical algorithm has cost $\mathcal{C}=O((Ld/\epsilon)^{(2d+L+1)/r}\max(L,d))$.
 \end{lemma}
 
  \begin{proof}
  See Appendix~\ref{app:classicaltransp} for details.
 \end{proof}
 
 \begin{corollary} \label{cor:transpadvantage1}
  To attain a quantum advantage when $M<M_{Boltz}$, it is  sufficient for the following condition to hold
\begin{align}\label{eq:transpadvantage1}
\littleo\left(\frac{Md^{1+(d-2)/r}}{n^2_{F_0}\Lambda L^{1+(3+d)/r}(d+L)^3}\left(\frac{1}{\epsilon}\right)^{(d-2)/r-1}\right)=\tilde{O}(1). 
\end{align}
When $M>M_{Boltz}$, it is sufficient that the following is satisfied 
\begin{align} \label{eq:transpadvantage2}
\littleo\left(\frac{\max(L,d)}{n^2_{F_0}\Lambda (L/\epsilon)(d+L)^3} \left(\frac{Ld}{\epsilon}\right)^{(d+L-2)/r}\right)=\tilde{O}(1). 
\end{align}
 \end{corollary}
 
 One can see from Eq.~\eqref{eq:transpadvantage1} that quantum advantage with respect to $M$ is always possible in the range $M'_{Boltz}<M<M_{Boltz}$ where $M'_{Boltz}=\tilde{\mathcal{O}}(n^2_{F_0}\Lambda L^{1+(3+d)/r}(d+L)^3\epsilon^{(d-2)/r-1}/d^{(d-2)/r+1})$, with up to exponential quantum advantage in $d$ and $\epsilon$. As a simple example, if we begin with a point source initial condition and only require the final solution at $O(1)$ points, then $n^2_{F_0}=O(1)=\Lambda$. When $d>2+r$, the range of possible $M$, captured by $M_{Boltz}-M'_{Boltz}$, is very large. A necessary condition for such an $M$ to exist in more general cases is $O(1)<n^2_{F_0}\Lambda< \tilde{\mathcal{O}}(L^{(d+L-2)/r-1}d(d/\epsilon)^{(d+L-2)/r-1}\max(L,d)/(d+L)^3)$, which allows for a wide range of possibilities. 

In the case when $M>M_{Boltz}$, we see from Eq.~\eqref{eq:transpadvantage2} there is no quantum advantage with respect to $M$, but there is quantum advantage in $L, d, \epsilon$ when $n^2_{V_0}\Lambda<\tilde{\mathcal{O}}((Ld/\epsilon)^{(d+L-2)/r}\max(L,d)/((L/\epsilon)(d+L)^3))$. Depending on $n^2_{V_0}\Lambda$, up to exponential quantum advantage in $L, d, \epsilon$ is possible. 


\subsection{Linear advection equation}

Now consider the linear advection equation
\begin{equation}\label{wave-eqn}
\begin{cases}
\partial_t u \pm a(x, z) \sum_{j=1}^d \partial_{x_j} u = 0,\\
u(0,x, z)=u_0(x,z).
\end{cases}
\end{equation}
Here $a(x,z)$ in uncertain due to the heterogeneity of the media \cite{mishra2016multi}.
For the case of $L=1$, applying the transformation $U(t,x,z,p)= a(z)e^{-p a(z)}u(t,x,z)$ as before, one can easily show that $U$ solves
 \begin{equation}\label{phase-Schr-adv}
\begin{cases}
\partial_t U \mp  \sum_{j=1}^d \partial_{x_j p} \,U = 0, \\
U(0,x,z,p)= a(z) e^{-p a(z)} u(0,x, z),
\end{cases}
\end{equation}
where $U$ satisfies a second order equation in the phase space. However, this problem is ill-posed,   which can be easily seen by applying a Fourier transform on $x$. Denote the Fourier variable  by $\xi\in \mathbb{R}^d$ in $x$ and $\iota=\sqrt{i}$. Then in the Fourier space the equation becomes
 \begin{equation}
     \partial_t{\hat{U}} \pm \iota \left(\sum_{j=1}^d\xi_j\right)  \partial_p{\hat{U}}=0
 \end{equation}
This is a convection equation with imaginary wave speed, which is ill-posed.

Instead we introduce a new transformation
\begin{equation} \label{Big-U-wave-1}
U(t,x,z,p)=\prod_{i=1}^L \sqrt{a_i(z)}\, e^{-\sqrt{a_i(z)} p_i}u(t,x,z)
 \end{equation}
and one arrives at the following problem:
\begin{equation}\label{phase-wave-0}
\begin{cases}
\partial_t U \pm  \sum_{i=1}^L b_i(x) \sum_{j=1}^d\partial_{x_jp_ip_i}\, U = 0, \\
U(0,x,p)=  \prod_{i=1}^L\sqrt{a_i(z)} e^{-p_i \sqrt{a_i(z)}} u(0,x, z),\\
\partial_{p_i}U(t,x,z,p)+\sqrt{a_m(z)}\,U(t,x,z,p)\Big|_{p_i=0}=0, \qquad {\text {for all}} \, i\,,
\end{cases}
\end{equation}
where is a {\it third} order equation in the phase space. 

 The well-posedness of the initial value problem can be seen by applying a Fourier transform on  $x$, for the case of constant $b_i(x)$. Assume the Fourier variable  is $\xi=(\xi_1, \cdots, \xi_d)^T\in \mathbb{R}^d$. Then taking a Fourier transform for $U$ on  $x$ gives
 \begin{equation}
     \partial_t{\hat{U}} \mp \iota \sum_{i=1}^L b_i  \left(\sum_{j=1}^d\xi_j\right) \, \partial_{p_ip_i} {\hat{U}}=0,
 \end{equation}
This is the free Schr\"odinger equation which is a good equation (well-posed)!

However, since the Robin boundary condition in Eq.~\eqref{phase-wave-0} still depends on $a(z)$, one will not be able to linearly superimpose solutions for different sample $z_m$ to obtain the ensemble average.  To remedy this problem, 
we 
introduce a new transformation
\begin{equation}
 {W}(t,x,p,q)=\frac{1}{M}\sum_{m=1}^M \prod_{i=1}^L \sqrt{a_i(z_m)}\, e^{-\sqrt{a_i(z_m)} (p_i+q_i)}u(t,x,z_m)   
\end{equation}
where $q\in \mathbb{R}^L$, which then solves
 \begin{equation}\label{phase-wave-1}
\begin{cases}
\partial_t {W} \pm  \sum_{i=1}^L b_i(x) \sum_{j=1}^d\partial_{x_jp_ip_i}\, {W} = 0, \\
{W}(0,x,p,q)= \frac{1}{M}\sum_{m=1}^M \prod_{i=1}^L \sqrt{a_i(z_m)} e^{-(p_i+q_i) \sqrt{a_i(z_m)}} u(0,x, z_m),\\
\partial_{p_i} {W}-\partial_{q_i} {W}\Big|_{p_i=0}=0 \quad {\text {for all}} \,\, i\,.
\end{cases}
\end{equation}
Now the coefficients in both the equation and the boundary condition no longer depend  on $a(z)$!

The average of the solutions of the original problem, Eq.~\eqref{wave-eqn}, can be recovered from $W(t,x,p,q)$ using
\begin{align}\label{-uaverage-11-adv}
\overline{u}(t,x)&=\frac{1}{M}\sum_{m=1}^M u(t,x, z_m)
= \int\int {W}(t, x, p, q)\, dp \,dq\,.
 \end{align}
 
 Due to the exponential decay in $p$ and $q$ we also use the homogeneous boundary condition at the right hand side of
the domain for $p$ and $q$ in a suitably truncated computational domain.
\\

{\bf {Remark}:}
Unlike for the linear heat equation and Boltzmann equations, the transformation where we extend the domain of $p_i$ to $(-\infty, \infty)$ does not work for the advection equation. For instance, if one uses $U(t,x,z,p)= (1/2)\sqrt{a(z)}\, e^{-\sqrt{a(z)} |p|}u(t,x,z)$, then one transforms the advection equation $\partial_t u+a(z)
\partial_x u=0$ into $ \partial_t U + \partial_{xpp} U = -\sqrt{a(z)}\delta(p) \partial_x U$, which clearly still has $a(z)$ dependence in its coefficient.\\

{\bf Remark:} 
To compute the variance, we can multiply \eqref{wave-eqn} by $u$ to derive
\begin{equation}\label{wave-eqn-2}
\begin{cases}
\partial_t u^2 \pm \frac{1}{2}a(x, z) \sum_{j=1}^d \partial_{x_j} u^2 = 0,\\
u^2(0,x, z)=u^2_0(x,z).
\end{cases}
\end{equation}
Once can then solve this equation similarly as the original advection equation to obtain the ensemble average of $u^2$.\\

Among the quantities of physical interest include $|u|$ (amplitude),  $u^2$ (energy) and $\nabla u$ (flux). To compute the ensemble average of the amplitude and energy, we can use 
\begin{align}
    \bar{u}(t,x)=\frac{1}{M}\sum_{m=1}^M u(t,x,z_m)=\int\int {W}(t,x,p,q) \,dp\, dq \approx \bar{u}_{n, j} \equiv \frac{1}{N_p^L N_q^L}\sum_{k}^{N_p^L }\sum_{l}^{N_q^L } W_{n, j, k,l}
\end{align}    
where $W_{n, j, k, l}=W(t=n\Delta t, x=j\Delta x, p=k\Delta p, q=l\Delta q)$.  The discretization meshes and corresponding notations in $x, t$ are the same as for the linear heat equation. The discretisation of $p$ and $q$, in one dimension for example, is $p_k=k\Delta p$ and $q_l=l\Delta q$ for $k=0,...,N_p$, $l=0,...,N_q$, where $N_p\Delta p=1=N_q\Delta q$. In $d$ spatial dimensions, the total number of spatial mesh points for $N_p$ and $N_q$ are respectively $N_p^L$ and $N_q^L$ since there are both $L$-dimensional. We will use the forward Euler discretisation for time, center finite difference in $x$, and upwind scheme in $w$ (so the overall spatial error is of $O(1/N)$ and temporal error is of $O(\Delta t)$. See Appendix \ref{AppendixB-3} for details of the discretisation scheme and the corresponding matrix $\mathcal{M}$ for the problem.

Denote the solution vector of Eq.~\eqref{phase-wave-1} by $\vect{W}=\sum_n^{N_t}\sum_j^{N^d}\sum_k^{N_p^L}\sum_l^{N_q^L} W_{n,j,k,l}|n\rangle |j\rangle |k\rangle |l\rangle$. Then $\vect{W}^T \mathcal{G}^{density}\vect{W}=N_p^LN_q^L |\bar{u}_{N_t, J}|^2$, where $\mathcal{G}^{density}=|G^{density}_{N_t, J}\rangle \langle G^{density}_{N_t, J}|$ and $|G^{density}_{N_t, J}\rangle=(1/\sqrt{N_p^LN_q^L})\sum_k^{N_p^L}\sum_l^{N_q^L}|N_t\rangle |J\rangle |k\rangle |l\rangle$. The quantum state that embeds the initial condition of $\vect{W}$ with sparsity $\sigma_{W_0}$ is defined as 
\begin{align}
    |W_0\rangle=\frac{1}{\mathcal{N}_{W_0}}\sum_j^{N^d}\sum_k^{N_p^L}\sum_l^{N_q^L}W_{0,j,k,l}|0\rangle |j\rangle |k\rangle |l\rangle 
\end{align}
where the normalisation constant is $\mathcal{N}_{W_0}^2=\sum_j^{N^d}\sum_k^{N_p^L}\sum_l^{N_q^L}|W_{0,j,k,l}|^2$. 

\begin{lemma} \label{lem:W0constantadv}
The upper bound to the normalisation constant $n^2_{W_0} \equiv \mathcal{N}^2_{W_0}/(N_p^LN_q^L)$ has range $O(1) \leq n^2_{W_0} \leq O(N^d)$
\end{lemma}

\begin{proof}
First, similar to the proof in Lemma \ref{lem:V0constant},  $ \int dp \int dq |W(0, x, p, q)|^2\le O(1)$. Then using the quadrature rule one gets 
$\int dx \int dp \int dq |W(0, x, p, q)|^2 \approx (1/ N^d)\sum_j^{N^d} \int dp \int dq |W(0, x=j\Delta x, p, q)|^2 \leq O(1)$. Then
$O(1) \geq \int dx \int dp \int dq |W(0,x, p,q)|^2 \approx \mathcal{N}^2_{W_0}/(N^d N_p^LN_q^L)=n^2_{W_0}/N^d$. Thus $O(1) \leq n^2_{W_0} \leq O(N^d)$. However, the quantum inner product $\Upsilon$ is upper bounded by $1$, where $1 \geq \Upsilon=\vect{W}^T\mathcal{G}^{density}\vect{W}/\mathcal{N}^2_{W_0}=O(N_p^LN_q^L/\mathcal{N}^2_{W_0})=O(1/n^2_{W_0})$ since $\bar{u}^2(t,x)=O(1)$. 
\end{proof}

We can then proceed as before using Lemmas~\ref{lem:bslep2} and ~\ref{lem:initialprep} to find the total cost in computing the density $|\bar{u}_{N_t,J}|^2$ at $\Lambda$ mesh points.

\begin{theorem}

 A quantum algorithm that takes sparse access to $\mathcal{M}$ (using an $r^{\text{th}}$-order method with $r\geq 1$) is able to estimate the energy $|\bar{u}_{N_t, J}|^2$ at $\Lambda$ points to precision $\epsilon$ with an upper bound on the gate complexity
\begin{align} \label{eq:qquerytotaladvection}
\mathcal{Q}=\mathcal{O}\left(n_{W_0}^2 \Lambda (Ld)^3\left(\frac{Ld}{\epsilon}\right)^{1+9/r}\right)
\end{align}
with a smaller query complexity.
\end{theorem}
\begin{proof}
Here $\mathcal{M}$ has sparsity $s=O(Ld)$ and condition number $\kappa=O(LdN^3)=O(Ld(Ld/\epsilon)^{3/r})$ where $N=O((Ld/\epsilon)^{3/r})$ from Appendix~\ref{app:classicalwave}.  Then  Lemma~\ref{lem:bslep2} can be applied to show a query complexity $\tilde{\mathcal{O}}(s\kappa^3\mathcal{N}^2_{W_0}/(N_p^LN_q^L\epsilon))=\tilde{\mathcal{O}}(n^2_{W_0}(Ld)^3(Ld/\epsilon)^{1+9/r})$ with the same order of additional $2$-qubit gates, where $\epsilon$ is the error in $|\bar{u}_{N_t, J}|^2$. 
\end{proof}

 \begin{lemma} \label{lem:classicalwave}
 When $M<M_{adv} \equiv O(L^{(d+2L+3)/r+2}(d/\epsilon)^{(2L+2)/r})$ the classical algorithm for the problem has the cost  $\mathcal{C}=O(Md^2(d/\epsilon)^{(d+1)/r})$. When $M>M_{adv}$, the classical algorithm has cost $\mathcal{C}=O((Ld)^2(Ld/\epsilon)^{(d+2L+3)/r})$.
 \end{lemma}
  \begin{proof}
See Appendix~\ref{app:classicalwave} for details. 
 \end{proof}

 \begin{corollary} \label{cor:wave}
  To attain a quantum advantage $M<M_{adv}$, it is  sufficient for the following condition to hold
\begin{align} \label{eq:waveadvantage1}
\littleo\left(\frac{M}{\Lambda n^2_{W_0}L^{4+9/r} d}\left(\frac{d}{\epsilon}\right)^{(d-8)/r-1}\right)=\tilde{O}(1). 
\end{align}
When $M>M_{adv}$ it is sufficient for the following condition to hold 
\begin{align} \label{eq:waveadvantage2}
\littleo\left(\frac{1}{n^2_{W_0}\Lambda  Ld}\left(\frac{Ld}{\epsilon}\right)^{(d+2L-6)/r-1}\right)=\tilde{O}(1). 
\end{align}
 \end{corollary}
  One can see from Eq.~\eqref{eq:waveadvantage1} that quantum advantage with respect to $M$ is always possible in the range $M'_{adv}<M<M_{adv}$ where $M'_{adv}=\tilde{\mathcal{O}}(n^2_{W_0}\Lambda L^{4+9/r}d(\epsilon/d)^{(d-8)/r-1})$, with up to exponential quantum advantage in $d$ and $\epsilon$. As a simple example, if we begin with a point source initial condition and only require the final solution at $O(1)$ points, then $n^2_{W_0}=O(1)=\Lambda$. When $d>8+r$, the range of possible $M$, captured by $M_{adv}-M'_{adv}$, is very large. A necessary condition for such an $M$ to exist in more general cases is $O(1)<n^2_{W_0}\Lambda< \tilde{\mathcal{O}}(L^{(d+L-6)/r-2}(d/\epsilon)^{(d+L-6)/r-1}/d)$, which allows for a wide range of possibilities. 

In the case when $M>M_{adv}$, one sees from Eq.~\eqref{eq:waveadvantage2} there is no quantum advantage with respect to $M$, but there is quantum advantage in $L, d, \epsilon$ when $n^2_{W_0}\Lambda<\tilde{\mathcal{O}}((Ld/\epsilon)^{(d+2L-6)/r-1}/(Ld))$. Depending on $n^2_{W_0}\Lambda$, up to exponential quantum advantage in $L, d, \epsilon$ is possible. 
 

 \subsection {The Schr\"odinger equation}
 
Consider the following Schr\"odinger equation with uncertainty
 \begin{equation}\label{Schr}
\begin{cases} i \hbar\,\partial_t u =-\frac{\hbar^2}{2}  \Delta u
+a(x, z) u\,, \\
u(0,x, z)=u_0(x,z)\,,
\end{cases}
\end{equation}
where $u=u(t,x, z)\in \mathbb{C}$ is the complex-valued wave function, $a(x,z)>0$ given in \eqref{a-generalform} is the potential, which could be uncertain due to modeling or computational errors, or random media \cite{ryzhik1996transport}. $\hbar$ is Planck's constant. For $L=1$, if we use the transformation
\begin{equation} \label{Big-U-Schr}
  U(t, x,z, p) = a(z) e^{-p a(z)} u(t,x, z)
 \end{equation}
 where $p\ge 0$, then $U$ solves
 \begin{equation}\label{phase-Schr}
\begin{cases}  i \hbar\, \partial_t U =-\frac{\hbar^2}{2}  \Delta U
- b(x)\partial_p U\,, \\
U(0,x,z,w)= a(z) e^{-p a(z)} u(0,x, z)\,.
\end{cases}
\end{equation}
However, this problem is again ill-posed. For example if $b(x)$ is a constant, by taking a Fourier transform on both $x$ and $p$ and assume the Fourier variable is $\xi\in \mathbb{R}^d$ in $x$ and $\eta\in \mathbb{R}$ in $p$, then
 
 \begin{equation}
     i\hbar \partial_t\hat{\hat{U}} =\left[\frac{1}{2}\hbar^2|\xi|^2+i b\eta\right]   \hat{\hat{U}},
 \end{equation}
which gives
\begin{equation}
 \partial_t\hat{\hat{U}} =\left[-\frac{i}{2}\hbar|\xi|^2+\frac{\eta}{\hbar}\right]   \hat{\hat{U}}.
\end{equation}
Clearly this equation is {\it unstable} since $\eta$ may be positive.

Like for the advection equation, for the Schr\"odinger equation we also cannot use the transformation applied to the linear heat and Boltzmann equations by extending $p_i$ to $(-\infty, \infty)$. Like for the advection equation, we define
\begin{equation}
 {W}(t,x,p,q)=\frac{1}{M}\sum_{m=1}^M \prod_{i=1}^L \sqrt{a_i(z_m)}\, e^{-\sqrt{a_i(z_m)} (p_i+q_i)}u(t,x,z_m) \,.  
\end{equation}
Then $W$ satisfies
\begin{align} \label{phase-Schr-2}
    \begin{cases}
    i\hbar \partial_t {W}=-\frac{\hbar^2}{2}\Delta {W}+ \sum_{i=1}^L b_i(x) \partial_{p_ip_i} {W}, \\
    {W}(0,x,p,q)=\frac{1}{M}\sum_{m=1}^M \prod_{i=1}^L \sqrt{a_i(z_m)} e^{-(p_i+q_i)\sqrt{a_i(z_m)}}u(0,x,z_m),\\
     \partial_{p_i} {W}-\partial_{q_i} {W}\Big|_{p_i=0} =0 \quad {\text {for all}} \,\, i\,.
    \end{cases}
\end{align}  
Note Eq.~\eqref {phase-Schr-2} is the free Schr\"odinger equation in the $(x,p)$ space, which is a good equation and the well-posedness of the initial value problem is classical. 
The average of the solutions of the original problem, Eq.~\eqref{Schr}, can be recovered from $W(t,x,p,q)$ using
\begin{align}\label{-uaverage-11}
\overline{u}(t,x)&=\frac{1}{M}\sum_{m=1}^M u(t,x, z_m)
= \int\int {W}(t, x, p, q)\, dp \,dq\,.
 \end{align}

The quantities of interest in real applications are the important physical observables of the original Schr\"odinger equation, which include position density $|u|^2$, flux or moment density $\hbar\,\text{Im} (\overline{u}\nabla u) $,  kinetic energy $\frac{1}{2}\hbar^2 |\nabla u|^2$, and total energy
$\frac{1}{2}\hbar^2 |\nabla u|^2+a(x,z)|u|^2$. \\

The simplest method to extract the ensemble average of the density $|\bar{u}(t,x)|^2=(1/M^2)|\sum_{m=1}^M u(t,x,z_m)|^2$ is from observing $ \bar{u}(t,x)=\int \int W(t,x,p,q) dp dq$. This integral can be approximated using the quadrature rule
\begin{align}
    |\bar{u}(t=N_t\Delta t, x=J \Delta x)| \approx |\bar{u}_{N_t, J}|=\frac{1}{N_p^LN_q^L}\sum_{k}^{N_p^L}\sum_l^{N_q^L} W_{N_t, J, k,l}
\end{align}
where $W_{N_t, J, k, l}=W(t=N_t \Delta t, x=J \Delta x, p=k\Delta p, q=l\Delta q)$. The discretization meshes and corresponding notations in $x, t$, $p$ and $q$ are the same as for the linear advection equation. We use the center difference scheme in $x$ and $p, q$, and forward Euler in $t$  to solve Eq.~\eqref{phase-Schr-2} to $t=1$. See Appendix \ref{AppendixB-4} for details of the discretisation scheme and the corresponding matrix $\mathcal{M}$ for the problem. 

Denote the solution vector of Eq.~\eqref{phase-Schr-2} by $\vect{W}=\sum_n^{N_t}\sum_j^{N^d}\sum_k^{N_p^L}\sum_l^{N_q^L}W_{n,j,k, l}|n\rangle |j\rangle |k\rangle |l\rangle$. Then $\vect{W}^T \mathcal{G}^{density}\vect{W}=N_p^L N_q^L|\bar{u}_{N_t, J}|^2$, where $\mathcal{G}^{density}=|G^{density}_{N_t, J}\rangle \langle G^{density}_{N_t, J}|$ and $|G^{density}_{N_t, J}\rangle=(1/\sqrt{N_p^LN_q^L})\sum_k^{N_p^L}\sum_l^{N_q^L}|N_t\rangle |J\rangle |k\rangle |l\rangle$. The quantum state that embeds the initial condition of $\vect{W}$ with sparsity $\sigma_{W_0}$ is defined as 
\begin{align}
    |W_0\rangle=\frac{1}{\mathcal{N}_{W_0}}\sum_j^{N^d}\sum_k^{N_p^L} \sum_l^{N_q^L}W_{0,j,k,l}|0\rangle |j\rangle |k\rangle |l\rangle 
\end{align}
where the normalisation constant is $\mathcal{N}_{W_0}^2=\sum_j^{N^d}\sum_k^{N_p^L}\sum_l^{N_q^L}|W_{0,j,k, l}|^2$. 

\begin{lemma} \label{lem:W0constant}
The upper bound to the normalisation constant $n^2_{W_0} \equiv \mathcal{N}^2_{W_0}/(N_p^LN_q^L)$ has range $O(1) \leq n^2_{W_0} \leq O(N^d)$
\end{lemma}

\begin{proof}
First, similar to the proof in Lemma \ref{lem:V0constant},  $\int \int dp dq |W(0, x, p, q)|^2 \le O(1) $. Then using the quadrature rule we have 
$\int dx \int dp \int dq |W(0, x, p, q)|^2 \approx (1/ N^d)\sum_j^{N^d} \int dp \int dq |W(0, x=j\Delta x, p, q)|^2 \leq O(1)$. Then
$O(1) \geq \int dx \int dp \int dq |W(0,x, p, q)|^2 \approx \mathcal{N}^2_{W_0}/(N^d N_p^LN_q^L)$. Thus $O(1) \leq \mathcal{N}^2_{W_0} \leq O(N_p^LN_q^L N^d)$.
\end{proof}

We can then proceed as before using Lemmas~\ref{lem:bslep2} and ~\ref{lem:initialprep} to find the total cost in computing the density $|\bar{u}_{N_t,J}|^2$ at $\Lambda$ mesh points.

\begin{theorem}
 A quantum algorithm that takes sparse access to $\mathcal{M}$ (using an $r^{\text{th}}$-order method with $r\geq 1$) is able to estimate the density $|\bar{u}_{N_t, J}|^2$ at $\Lambda$ points to precision $\epsilon$ with an upper bound on the gate complexity
\begin{align} \label{eq:qquerytotalschrodinger}
\mathcal{Q}=\tilde{\mathcal{O}}\left(n_{W_0}^2 \Lambda (d+L)^3\left(\frac{d+L}{\epsilon}\right)^{1+6/r}\right)
\end{align}
with a smaller query complexity.
\end{theorem}
\begin{proof}
Here $\mathcal{M}$ has sparsity $s=O(d+L)$ and condition number $\kappa=O((d+L)N^2)=O((d+L)((d+L)/\epsilon)^{2/r})$ where $N=O(((d+L)/\epsilon)^{1/r})$ from Appendix~\ref{app:classicalSchr}. Then  Lemma~\ref{lem:bslep2} can be applied to show a query complexity $\tilde{\mathcal{O}}(s \kappa^3\mathcal{N}^2_{W_0}/(N_p^LN_q^L\epsilon))=\tilde{\mathcal{O}}(n^2_{W_0}(d+L)^3((d+L)/\epsilon)^{1+6/r})$ with the same order of additional $2$-qubit gates, where $\epsilon$ is the error in $|\bar{u}_{N_t, J}|^2$. 
\end{proof}
 
  \begin{lemma} \label{lem:classicalschrodinger}
 When $M<M_{Schr} \equiv O((d+L)^{(d+2L+2)/r+2}/(d^{(d+2)/r+2}\epsilon^{2L/2}))$, the classical algorithm for the problem has the cost  $\mathcal{C}=O(Md^2(d/\epsilon)^{(d+2)/r})$. When $M>M_{Schr}$, the classical algorithm has cost $\mathcal{C}=O((d+L)^2((d+L)/\epsilon)^{(d+2L+2)/r})$.
 \end{lemma}
  \begin{proof}
See Appendix~\ref{app:classicalSchr} for details.
 \end{proof}

 \begin{corollary} \label{cor:schrodinger}
   To attain a quantum advantage when $M<M_{Schr}$, it is  sufficient for the following condition to hold
\begin{align} \label{eq:Schradvantage1}
\littleo\left(\frac{M d^{(d+2)/r+2}}{n^2_{W_0}\Lambda (d+L)^{4+6/r} \epsilon^{(d-4)/r-1}}\right)=\tilde{O}(1). 
\end{align}
When $M>M_{Schr}$ it is sufficient for the following to hold 
\begin{align} \label{eq:Schradvantage2}
\littleo\left(\frac{(d+L)^{(d+2L-4)/r-2}}{n^2_{W_0}\Lambda \epsilon^{(d+2L-4)/r-1}}\right)=\tilde{O}(1). 
\end{align}
 \end{corollary}
 
   One sees from Eq.~\eqref{eq:Schradvantage1} that quantum advantage with respect to $M$ is always possible in the range $M'_{Schr}<M<M_{Schr}$ where $M'_{Schr}=\tilde{\mathcal{O}}(n^2_{W_0}\Lambda (d+L)^{4+6/r}\epsilon^{(d-4)/r-1}/d^{(d+2)/r+2})$, with up to exponential quantum advantage in $d$ and $\epsilon$. As a simple example, if one begins with a point source initial condition and only requires the final solution at $O(1)$ points, then $n^2_{W_0}=O(1)=\Lambda$. When $d>c+4$, the range of possible $M$, captured by $M_{Schr}-M'_{Schr}$, is very large. A necessary condition for such an $M$ to exist in more general cases is $O(1)<n^2_{W_0}\Lambda< \tilde{\mathcal{O}}((d+L)^{(d+L-4)/r+2}/( \epsilon^{L/2+(d-4)/r}))$, which allows for a wide range of possibilities. 

In the case when $M>M_{Schr}$, one sees from Eq.~\eqref{eq:Schradvantage2} there is no quantum advantage with respect to $M$, but there is quantum advantage in $L, d, \epsilon$ when $n^2_{W_0}\Lambda<\tilde{\mathcal{O}}((d+L)^{(d+2L-4)/r-2}/(\epsilon^{(d+2L-4)/r-1}))$. Depending on $n^2_{W_0}\Lambda$, up to exponential quantum advantage in $L, d, \epsilon$ is possible. 
 
 
\section{Discussion}
We remark that it is also possible to compute ensemble averaged observables from the original PDEs directly with a corresponding quantum algorithm, where the PDE needs to be solved $M$ times for $M$ samples of $z$. In this case, however, different oracle assumptions are required. To solve the original PDE with $M$ initial conditions, there are $M$ different corresponding $\mathcal{M}$ matrices, since $a(x,z)$ appears explicitly in the PDE to to solved. Therefore $M$ different sparse oracles are required. This is in contrast to solving the phase space PDEs, where only a single oracle is needed, which is furthermore \textit{independent} of the details of the stochastic model $a(x,z)$! This can be relevant for scenarios where either the sparse oracles for $\mathcal{M}$ are not easily programmable to high precision or one is not given access to be able to tune the sparse oracles for different values of $z$. To solve the phase space problem, only a \textit{single} type of oracle to access $\mathcal{M}$ is required for each type of PDE, rather than a different oracle for every different instance of the PDE.

However, in cases that multiple sparse oracles for $\mathcal{M}$, one for each value of $z$, is assumed to be given without any cost, then solving the original PDE with a quantum algorithm with cost $\mathcal{Q}_{orig}$ is only advantageous compared to our new algorithms only for relatively smaller values of $M$, given in Table~\ref{table2}. For modest sizes of $M$, we see that our new quantum algorithms are still preferred, i.e. $\mathcal{Q}<\mathcal{Q}_{orig}$. See Appendix~\ref{app:qoriginal} for details.

We note that the conditions in Table~\ref{table1} and Table~\ref{table2} do not assume costs in preparing the initial quantum state. As discussed in Section~\ref{sec:PDE} it is straightforward to include the initial state preparation cost in Table~\ref{table1} by multiplying the quantum cost by the sparsity $\sigma_0$ of the initial state and to include another factor of $(d+L)$, following Lemma~\ref{lem:initialprep}. Here if we define the input-output factor (IOF) as $\Gamma \equiv \sigma_0 n_0^2 \Lambda=N^b$, then the minimum $b=L$, instead of $b=0$ in the current table. We see that this still leaves room for quantum advantage usually when $d$ is larger than $L$.

It is crucial to emphasise that here the initial state preparation costs (and hence the IOF) is \textit{independent} of $M$. Neither the sparsity, the dimension nor the normalisation of the initial quantum state depends on $M$. This is because each amplitude of the initial quantum state is already an average over $M$ terms associated with each initial condition, and this amplitude is first computed \textit{classically} before being embedded in a quantum state. This means that the end-to-end cost of our quantum algorithms in Table~\ref{table1} are all \textit{independent} of $M$. 

Suppose we now include initial state preparation costs in comparing $\mathcal{Q}$ with $\mathcal{Q}_{orig}$ in Table~\ref{table2}. We note that their respective IOF factors $\Gamma \equiv \sigma_0 n_0^2 \Lambda$ are different because the sparsities of their initial condition is different up to a factor of $N^L$. For instance, for a point source in the original problem, where the sparsity of of the initial state is $1$, the sparsity of the initial quantum state for the phase space equation is $N^L$. If we can assume that in solving the original PDE $M$ times we are given all the corresponding oracles for $\mathcal{M}$ for free -- which is not always possible -- and if we choose the same initial state preparation scheme in Lemma~\ref{lem:initialprep} as before, the quantum algorithm to solve the original PDE $M$ times is preferred unless we have very large $M$. Alternatively, if we choose a different initial preparation scheme like in \cite{zhang2022quantum}, then if we count only gate complexity, we only need to multiply the quantum cost by a logarithmic factor of $\sigma_0$, although order $\sigma_0(d+L)$ ancilla qubits are now necessary. In this case, the conditions in Table~\ref{table2} still hold up to an order of $\tilde{\mathcal{O}}(L(d+L))$ when including initial state preparation.

    \begin{table}[ht] 
\caption{Here $\mathcal{Q}_{orig}$ and $\mathcal{Q}$ are the respective quantum costs for $r^{\text{th}}$-order methods in computing ensemble averaged observables, over $M$ samples, based on solving the original equation  ($\mathcal{Q}_{orig}$) versus the phase space representation ($\mathcal{Q}$). Here $s$ and $\kappa$ are the sparsity and condition numbers corresponding to $\mathcal{M}$ of the original $(d+1)$-dimensional PDE and the phase space PDE. However, there are more oracle assumptions needed for $\mathcal{Q}_{orig}$ and when these assumptions are not obeyed, the phase space method is always preferable. See main text for discussion.} 
    
\begin{tabular}{c c c c c }
\hline\hline 
PDE & \qquad $s$ &  \qquad $\kappa$ & \qquad $\mathcal{Q}<\mathcal{Q}_{orig}$  \\[2ex]
\hline 
\\

Linear heat &\quad $d$ & \qquad $(d/\epsilon)^{2/r}$ &   \qquad $M>\tilde{\mathcal{O}}(L^{4+9/r}d^3(d/\epsilon)^{3/r})$ \\
Phase space linear heat  &\quad $Ld$ & \qquad $L d(Ld/\epsilon)^{3/r}$   \\[3ex]

Linear Boltzmann &\quad $(d/\epsilon)^{d/r}$ & \qquad $d(d/\epsilon)^{1/r}$ &  \qquad $M>\tilde{\mathcal{O}}(L^{1+(3+d)/r}(d/\epsilon)^{2/r}/(d\epsilon))$ \\
Phase space linear Boltzmann &\quad $L(Ld/\epsilon)^{d/r}$ & \qquad $(d+L)(Ld/\epsilon)^{1/r}$  \\ [3ex]

Linear advection  &\quad $d$ & \qquad $d(d/\epsilon)^{1/r}$ &  \qquad $M>\tilde{\mathcal{O}}(L^{4+9/r}d^2(d/\epsilon)^{8/r})$ \\
Phase space linear advection &\quad $Ld$ & \qquad $Ld(Ld/\epsilon)^{3/r}$ &   \\ [3ex]

Schr\"odinger &\quad $d$ & \qquad $(d/\epsilon)^{2/r}$ &  \qquad $M>\tilde{\mathcal{O}}((d+L)^3(1+L/d)^{1+2/r}((L+d)/\epsilon)^{4/r})$ \\
Phase space Schr\"odinger & \quad $d+L$ & \qquad $(d+L)((d+L)/\epsilon)^{2/r}$ &  \\ [3ex]

\hline 
\end{tabular} \label{table2} 
\end{table}

We observe that while classical random walk methods can also in principle be used to solve the phase space versions of these PDEs and these algorithms do not suffer from the curse of dimensionality, we don't choose to compare our quantum algorithms to these classical costs. Firstly, quantum advantage compared to classical random walk methods are still possible, albeit are at most polynomial and not exponential, for example \cite{linden2020quantum}. However, only advantage in $d$, $\epsilon$ are usually discussed, whereas here our main focus is on a different type of advantage with respect to $M$ and $L$, while still having potential advantage in $d$ and $\epsilon$. Secondly, it is well-known that classical random walk methods can be prone to a lot of statistical noise and also have a lower order of accuracy. We consider it more appropriate therefore to compare the quantum algorithm with its direct classical counterpart, i.e. using finite difference methods, which can offer highly accurate and high resolution solutions if competitive high-order finite difference methods are used.   
\section{Conclusion}
We have introduced new quantum algorithms that can compute ensemble averaged observables of PDEs with uncertainty. The end-to-end cost of these algorithms (including initial state preparation and final measurement) can be \textit{independent} of the number of initial conditions $M$ and in certain regimes can also show advantage in $L$, $d$ and $\epsilon$. 

The key idea of ours is to develop new transformations that map PDEs with uncertainty to deterministic PDEs. We transform the original PDEs with $M$ initial conditions into phase space PDEs in higher dimensions, so the transformed equations are \textit{independent} of the uncertain coefficients with a \textit{single} initial condition. Although solving these new PDEs classically -- while also guaranteeing high accuracy and resolution -- would be inefficient in $d$ and $\epsilon$, quantum algorithms are efficient in $d$, $\epsilon$ and also provide advantage in $M$ and $L$ in different regimes. 

These transformations can also be extended to certain nonlinear PDEs using techniques in \cite{jin-liu-2022}, which introduced another way of embedding the ensemble average over $M$ in the initial condition--to deal with the nonlinearity rather than the uncertainty in this case. This will be explored in future work. 

These new algorithms offer exciting opportunities to begin exploring how quantum methods can be applied to uncertainty quantification. 

\section*{Acknowledgement}
SJ was partially supported by the NSFC grant No.~12031013, the Shanghai Municipal Science and Technology Major Project (2021SHZDZX0102), and the Innovation Program of Shanghai Municipal Education Commission (No. 2021-01-07-00-02-E00087).  NL acknowledges funding from the Science and Technology Program of Shanghai, China (21JC1402900), the Shanghai Pujiang Talent Grant (no. 20PJ1408400) and the NSFC International Young Scientists Project (no. 12050410230).

\section*{Appendix}
 \appendix 
  
\section{Numerical discretisation of PDEs, $\mathcal{M}$ matrices and their properties} \label{app:numerics}

\subsection{Linear heat equation}\label{AppendixB-1}

\subsubsection{Numerical discretizations}

As an example  we use center finite difference scheme to solve the phase space heat equation \eqref {phase-heat-11}. For convenience we use homogeneous boundary condition in $x$, and the case $b_i(x)>0$, so  the zero boundary condition in $u$ is given on the right boundary. We also just consider $p>0$, since the inclusion of the case of  $p<0$ will not affect the computation of condition number and sparsity. We first consider spatially one-dimensional (for both $x$ and $w$) problem. 
Consider $N_t+1$ steps in time $0 = t_0<t_1<\cdots<t_{N_t} = 1$ and $N_x+1$ spatial mesh points $0<x_0<x_1<\cdots<x_{N_x} = 1$, $0<p_0<p_1<\cdots<p_{N_x} = 1$  by setting $t_n = n\tau$ and $x_j =p_j= jh$, where $\tau = 1/N_t$ and $h = 1/N_x$. Let $u^n_{jk}$ denote the numerical approximation of $u$ at $(t_n, x_j, p_k)$. We use the forward Euler method in time,  center difference in $x$, and upwind scheme in $p$ (so the overall spatial error is of $O(h)$):
\begin{eqnarray}
 && \frac{u^{n+1}_{jk} - u^n_{jk}}{\tau} 
  +\frac{1}{h^3} \left[
  (u_{j-1,k+1}^{n} - 2u_{j,k+1}^{n} + u_{j+1,k+1}^{n})
  -(u_{j-1,k}^{n} - 2u_{j,k}^{n} + u_{j+1,k}^{n})\right]=0,
  \label{CN-heat}\\
 && \qquad \qquad \qquad \qquad \qquad \qquad 
    j= 1,\cdots, N_x-1; \quad k=0,\cdots, N_x-1.\nonumber 
\end{eqnarray}


Define
\[L_h =
\begin{bmatrix}
-2  &  1       &           &      &    \\
 1  & -2       & \ddots    &      &    \\
    &  \ddots  & \ddots    &  \ddots    &    \\
    &          & \ddots    & -2   & 1  \\
    &          &           &  1   & -2 \\
\end{bmatrix}_{(N_x-1) \times (N_x-1)}
, \qquad
b(k) =
\begin{bmatrix}
u_0(k+1)-u_0(k) \\
0  \\
\vdots\\
0 \\
u_{N_x}(k+1)-u_{N_x}(k) \\
\end{bmatrix},
\]
and set $\lambda=\tau/h^3$, $u_k^n=(u_{1,k}^n, \cdots, u_{N_x,k}^n)^T$, then \eqref{CN-heat} can be written as
\begin{equation}\label{CN-vector}
 u_k^{n+1}
-\left(I+\lambda L_h\right) u_k^{n} 
+\lambda L_h u_{k+1}^{n}
=  \lambda b(k),
\end{equation}
which can be marched forward in time by solving a linear system for  a classical computer.

Let $u^n=(u_0^n, \cdots, u_{N_x-1}^n)^T$. Using the boundary condition $u^n_{N_x}=0$,  then \eqref{CN-vector} can be further written as
\begin{equation}\label{heatC}
-Bu^n +  u^{n+1} = f^{n+1}=:\lambda(b_0, \cdots, b_{N_x-1})^T ,
\end{equation}
where
\[
B =
\begin{bmatrix}
I+\lambda L_h  & - \lambda L_h          &           &        &    \\
 & I+\lambda L_h     &  -\lambda L_h         &         &   \\
& & \ddots&\ddots &\\
     &   &      & \ddots    &  -\lambda L_h  \\
      &  &            &   & I+\lambda L_h     \\
\end{bmatrix}.
\]
By introducing the notation $U = [u^1, \cdots ,u^{N_t}]^T$,  one  obtains the following linear system
\begin{equation}\label{heatLF}
L U = F,
\end{equation}
where
\[L =
\begin{bmatrix}
I  &            &           &            \\
-B & I     &           &            \\
        &\ddots      & \ddots    &    \\
        &            & -B   & I     \\
\end{bmatrix}
, \qquad
F =
\begin{bmatrix}
f^1 + Bu^0 \\
f^2  \\
\vdots\\
f^{N_t} \\
\end{bmatrix}.
\]
Note that the $\mathcal{M}$ matrix is denoted
\[H = \begin{bmatrix} O & L \\  L^T  & O  \end{bmatrix}.\]

\subsubsection{Estimation of the eigenvalues}

Let $\lambda$ be an eigenvalue of $H$, i.e.,
\[\mbox{det}(\lambda I-H) = \mbox{det}(\lambda^2I-LL^T) =:\mbox{det}(\mu I-LL^T),\]
where $\mu = \lambda^2$ is the eigenvalue of $LL^T$, or $\mu^{1/2}$ is the singular value of $L$. A direct calculation gives
\[LL^T =
\begin{bmatrix}
I  &    -B             &                   &            \\
-B  & I+B^2      & \ddots            &            \\
          &\ddots                   & \ddots            &   -B \\
          &                         & -B          & I+B^2    \\
\end{bmatrix}.
\]

We  establish the upper and lower bounds of the eigenvalues of $LL^T$ by using the Gershgorin circle theorem.

\begin{lemma}\label{lem:heatEig}
Let $\lambda = \tau/h^3$. Then the minimum and maximum eigenvalues of   $H$ satisfy
\[|\lambda|_{\min} \ge  \tau , \qquad |\lambda|_{\max} \le 2. 
\]
Consequently, for matrix $H$,  $s=7, \kappa \le \frac{2}{\tau} =O( N_t) =O( N_x^3)$.
\end{lemma}

\begin{proof}
Let $P^{-1}L_hP = \Lambda$, where $\Lambda$ is the diagonal matrix consisting of the eigenvalues of $L_h$.
Then 
\[
  B \sim \Pi_B:=
\begin{bmatrix}
I+\lambda\Lambda  & - \lambda\Lambda         &           &        &    \\
 & I+\lambda\Lambda     & -\lambda   \Lambda      &         &   \\
& & \ddots&\ddots &\\
     &   &      & \ddots    &  -\lambda \Lambda \\
      &  &            &   & I+\lambda  \Lambda   \\
\end{bmatrix}
\]
and
\[
B^2 \sim \Pi_{B^2}=
\begin{bmatrix}
(I+\lambda\Lambda)^2 &  -2\lambda\Lambda(I+\lambda\Lambda)    &   \lambda^2\Lambda^2      &                   &\\

          &  \ddots                & \ddots    &\ddots        &  \\
&& (I+\lambda\Lambda)^2 & - 2\lambda\Lambda(I+\lambda\Lambda)    &   \lambda^2\Lambda^2     \\
&& &(I+\lambda\Lambda)^2 &  -2\lambda\Lambda(I+\lambda\Lambda)    &       \\
 &                         & &   & (I+\lambda\Lambda)^2    \\
\end{bmatrix}. \\
\]
Let $\widetilde{P} = \mbox{diag}(P, \cdots, P)$. Then
\[
\widetilde{P}^{-1}(LL^T)\widetilde{P} =
\begin{bmatrix}
I  &    -\Pi_{B}             &     &              &            \\
-\Pi_{B}  & I + \Pi_{B^2}     &  -\Pi_{B}            &       &     \\
         &\ddots                   & \ddots              \ddots &   \\
         &                  & \ddots              \ddots & -\Pi_{B}  \\
        &                         & -\Pi_{B}          & I + \Pi_{B^2}    \\
\end{bmatrix}.
\]

The eigenvalues of $L_h$ are
\begin{equation}\label{eigLh}
\nu_{l} = -4\sin^2\frac{l\pi}{2N_x} = -4\sin^2\frac{l\pi h}{2}, \quad l = 1,\cdots, N_x-1
\end{equation}
Noting that the similarity transformation does not change the eigenvalues, thus one can apply the Gershgorin circle theorem to estimate the  maximum eigenvalues of this matrix. In particular,
\begin{align}\label{numax}
\mu_{\max} \le \max_{l} \{1+ (1+\lambda  \nu_l)^2 + 2(|1+\lambda\nu_l|+\lambda |\nu_l|)|+2\lambda |\nu_l||1+\lambda \nu_l|+\lambda^2\nu_l^2 \}\le 4,
\end{align}
where we need the 
assumption $\lambda\le 1/4$.  

By definition,
\[\sigma_{\min}(L) = \frac{1}{\sigma_{\max}(L^{-1})}.\]
By simple algebra, one has
\begin{equation}
L^{-1} =
\begin{bmatrix}
I          &              &                 &          &            \\
B        & I      &                 &          &        \\
 B^2      &   B        & I       &                  &\\
   \vdots       &  \ddots      &  \ddots         &  \ddots  &   \\
B^{N_t-1} &   \cdots     &       &  B &  I      \\
\end{bmatrix}
=\begin{bmatrix}
I          &              &                 &          &            \\
                & I       &                 &          &        \\
                &              &  \ddots         &          &\\
                &              &                 &  \ddots  &   \\
                &              &                 &          &  I      \\
\end{bmatrix}+
\begin{bmatrix}
                 &              &                 &          &            \\
B        &              &                 &          &        \\
                 &  B      &                 &          &\\
                 &              &  \ddots         &          &   \\
                 &              &                 & B
\end{bmatrix} + \cdots, 
\end{equation}
which gives
\begin{align*}
\sigma_{\max}(L^{-1})
& = \|L^{-1}\|_2 \le \|I\|_2 + \|B\|_2 + \cdots +\|B^{N_t-1}\|_2 \\
& \le \|I\|_2 + \|B\|_2 + \|B\|_2^2 +\cdots +\|B\|_2^{N_t-1} ,
\end{align*}
Moreover,
\[
BB^T \sim \Pi_{B_2}=
\begin{bmatrix}
(I+\lambda\Lambda)^2+\lambda^2\Lambda^2 &  -\lambda\Lambda(I+\lambda\Lambda)    &       &                   &\\
-\lambda\Lambda(I+\lambda\Lambda)& (I+\lambda\Lambda)^2+\lambda^2\Lambda^2 &  -\lambda\Lambda(I+\lambda\Lambda) &&\\
          &  \ddots                & \ddots    &\ddots        &  \\
&& -\lambda\Lambda(I+\lambda\Lambda)& (I+\lambda\Lambda)^2+\lambda^2\Lambda^2 &  -\lambda\Lambda(I+\lambda\Lambda)    \\
&&& -\lambda\Lambda(I+\lambda\Lambda)& (I+\lambda\Lambda)^2+\lambda^2\Lambda^2  \\
 &                         & &   & (I+\lambda\Lambda)^2+\lambda^2\Lambda^2    \\
\end{bmatrix}.
\]
Since $\|B\|_2$ is the square root of the maximum eigenvalue of $BB^T$, again by Gershgorin theorem,
\[
\|B\|_2^2\le \max_l\{ (1+ \lambda \nu_l)^2 + \lambda^2\nu_l^2+ 2\lambda|\nu_l||1+\lambda\nu_l|\le 1 
\]
 by using $\lambda\le 1/4$. Hence,
\begin{align*}
\sigma_{\max}(L^{-1}) \le  N_t=1/\tau.
\end{align*}
Thus, $\sigma_{\min}(L) \ge \tau$.

To conclude, the matrix $H$ has order $N=2N_t(N_x-1)^2N_x$ and sparsity number $s \le 7$. From the above  eigenvalue estimations,  one knows that the condition number
\[\kappa \le \frac{2}{\tau} \sim N_t \sim N_x^3\]
using the condition $\lambda=\tau/h^3$.
\end{proof}

%

\subsubsection{Higher space dimension}

Now we consider the case of $x\in  [0,1]^d$,  with homogeneous Dirichlet boundary condition in $x$,  Other boundary conditions can be similarly treated with possible different $\kappa$. Let $p\in [0,1]^L$,  with zero Dirichlet boundary condition at $p_i=1$ for all $i$. We will use a uniform mesh size $h$ with the same number of mesh points $N_x$ in all space dimensions of $x$ and $p$.

Let $D=dL$. The matrix $L_h$ of the $D$-dimensional problem will be replaced by
\[
L_{h,D} = \underbrace{L_h\otimes I \otimes \cdots \otimes I}_{D~\mbox{matrices}} + I \otimes L_h\otimes \cdots \otimes I + \cdots + I \otimes I \otimes \cdots \otimes L_h,
\]
and everything else remains the same. We can still carry out the similar proof in Lemma \ref{lem:heatEig}, except that the eigenvalues $\nu_{l}$ of $L_{h}$ will replaced by the corresponding results of the new matrix. By the properties of tensor products, the eigenvalues of $L_{h,D}$ can be represented by the sum of the eigenvalues of $L_h$ as
\[
\nu^D_{i_1,i_2,\cdots, i_D} = \nu_{i_1} + \nu_{i_2} + \cdots + \nu_{i_D}, \quad  \nu_{l} =-4 \sin^2(l\pi h/2)  \]
for $i_l=1, \cdots, N_x-1, l=1, \cdots, D$.
Like was done in \eqref{numax}, one gets
\begin{eqnarray}
\nonumber
&&\mu^D_{\max} \le \max_{l} \{1+ (1+\lambda \nu^D_{i_1,i_2,\cdots, i_D} )^2 + 2(|1+\lambda\nu^D_{i_1,i_2,\cdots, i_D}|+\lambda |\nu^D_{i_1,i_2,\cdots, i_D}|)|+2\lambda |\nu^D_{i_1,i_2,\cdots, i_D}||1+\lambda \nu^D_{i_1,i_2,\cdots, i_D}|
\\
&&\qquad +\lambda^2(\nu^D_{i_1,i_2,\cdots, i_D})^2 \}\le 4
\end{eqnarray}
if 
\begin{equation}
    \lambda=\tau/h^3 \le 1/(4D).
\end{equation}
Like what was done in Lemma \ref{lem:heatEig}, $\nu^D_{max}(L^{-1})=O(N_t)=O(DN_x^3)$.  Hence 
$\kappa=O(DN_x^3)=O(Ld N_x^3)$.

One can easily check that the sparsity $s=O(Ld)$.

\subsection{Linear Boltzmann equation}\label{AppendixB-2}

We now consider the numerical approximation of the phase space linear Boltzmann equation with isotropic scattering in Eq.~\eqref{phase-Transp-11}. 
For simplicity we only consider homogeneous boundary condition in $x$ and the case of $b_i(x)>0$ so only the homogeneous boundary condition on the right boundary of $p$ is needed.
We also just consider $p_i\in [0,1]$. 

We start with one-dimensional $x\in [0,1], v\in (-1,1)$ and $p\in [0,1]$. Then
Eq.~\eqref{phase-Transp-11} is
\begin{equation}\label{phase-Transp-1d}
\begin{cases} \partial_t F + v\partial_x F = - \left[
\frac{1}{2}\int_0^1 \partial_p F\, dv - \partial_p F\right]\,, \\
F(0,x,v,z,p)= a(z) e^{-p a(z)}f_0(x,v,z).
\end{cases}\, \\
\end{equation} 
We use the upwind scheme for $x$ and $p$, and discrete-ordinate method for $v$, with $(\omega_m, v_m), m=\pm 1, \cdots, \pm M$ the pair of quadrature points and weights to approximate the integration in $v$. Here we first define quadrature points on $(0,1)$ and then use symmetry to define the quadrature points over $(-1,0)$, therefore $v_{-m}=-v_m, \omega_{-m}=\omega_m$ and $0$ is not a quadrature point. The weights $\omega_m$ satisfy
\begin{equation}\label{w-cond}
  \frac{1}{2}\sum_{|m=1|}^M \omega_m =1\,.
 \end{equation}
Let $F^n_{j,m, l,k}$ be the approximation of $F(t,x,v,z,p)$ at $(t^n, x_j, v_m, z_l,p_k)$. Then the discrete system of \eqref{phase-Transp-1d} is
\begin{eqnarray}\label{Transp-discrete}
\begin{cases} 
& \frac{1}{\tau}(F^{n+1}_{j,m,l,k}-F^{n}_{j,m,l,k}) + \frac{1}{2h}(v_m+|v_m|) (F^{n}_{j,m,l,k}-F^{n}_{j-1,m,l,k})
+ \frac{1}{2h}(v_m-|v_m|) (F^{n}_{j+1,m,l,k}-F^{n}_{j,m,l,k})\\
 &   = -\frac{1}{h}\frac{1}{2}\sum_{|m'|=1}^M \omega_{m'} (F^{n}_{j,m',l,k+1}-F^{n}_{j-1,m',l,k})
 +\frac{1}{h} (F^{n}_{j,m,l,k+1}-F^{n}_{j-1,m,l, k})\, \\
 &  \\
& F^0_{j,m,l,k}= a(z_l) e^{-p_k a(z_l)}f_0(x_j,v_m,z_l)
\end{cases}\, 
\end{eqnarray}
 
Since the evolution equation does not depend on $z$, we will drop the index  $l$ from $F_{j,m,l,k}$ in the sequel as long as there is no confusion.

Introduce $f_{j,k}=(f_{j,-M,k}, \cdots, f_{j,M,k})^T$. Define
\[
V^-=
\begin{bmatrix}
v_{-M}&&&&&\\
& \ddots &&&&\\
&& v_{-1} &&&\\
&&&0&&\\
&&&&\ddots&\\
&&&&&0
\end{bmatrix},
\qquad 
V^+=
\begin{bmatrix}
0&&&&&\\
& \ddots &&&&\\
&&0 &&&\\
&&&v_1&\\
&&&&\ddots&\\
&&&&&v_M
\end{bmatrix},
\qquad
I_M=
\begin{bmatrix}
1& \cdots & 1\\
\cdots & \cdots & \cdots\\
1& \cdots & 1
\end{bmatrix}_{2M\times 2M}.
\]
and $W=\frac{1}{2}\text{dial}(\omega_{-M}, \cdots, \omega_M)^T$.  Then scheme \eqref{Transp-discrete}
can be written in vector form as
\begin{equation}
    f^{n+1}_{j,k}-f^{n}_{j,k}
    +\lambda {v^{-}} \,(f^{n}_{j+1,k}-f^{n}_{j,k})+ \lambda {v^{+}} \,(f^{n}_{j,k}-f^{n}_{j-1,k})
    - \lambda I_M W (f^{n}_{j,k+1}-f^{n}_{j,k})+ \lambda (f^{n}_{j,k+1}-f^{n}_{j,k})=0,
\end{equation}
where $\lambda = \tau/h$. Denote $f^n=(f^n_{1}, \cdots, f^n_{N_x-1})^T$, and
\[
L_h=
\begin{bmatrix}
-I & I &&&\\
& -I & I &&\\
&& \ddots& \ddots&\\
&&&-I&I
\end{bmatrix}_{2M(N_x-1)\times 2M(N_x-1)}\qquad
\mathbb{V}^\pm = {\text{diag}}(V^\pm, \cdots, V^\pm), \qquad
\mathbb{W} = {\text{diag}}(I_MW,  \cdots, I_M W)
\]
where $I$ is the $2M \times 2M$ identity matrix,
and $\mathbb{I}$ is the $2M \times (N_x-1)$ identity matrix. Then the above scheme
can be written as
\begin{equation}
    f^{n+1}_{k}-f^{n}_{k}
    +\lambda \mathbb{V}^{-} L_h \,f^{n}_{k}+ \lambda \mathbb{V}^{+} \,L_h^Tf^{n}_{k}
    - \lambda  \mathbb{W}  (f^{n}_{k+1}-f^{n}_{k})+ \lambda (f^{n}_{k+1}-f^{n}_{k})=\lambda b(k)
\end{equation}
with $b(k)=(v^+f_{0,k}, \cdots, -v^- f_{N_x, k}) $.

Let $u^n=(u_0^n, \cdots, u_{N_x-1}^n)^T$. Using the boundary condition $u^n_{N_x}=0$,  then \eqref{CN-vector} can be further written as
\begin{equation}\label{heatCBoltzmann}
-Bu^n +  u^{n+1} = f^{n+1}=:\lambda(b_0, \cdots, b_{N_x-1})^T ,
\end{equation}
where
\[
B =
\begin{bmatrix}
B_1 &   B_2      &           &        &    \\
 &  B_1 &   B_2        &         &     \\
& & \ddots&\ddots &\\
     &   &      & \ddots    &  B_2  \\
      &  &            &   &  B_1  \\
\end{bmatrix},
\]
with $B_1=I-\lambda \mathbb{V}^{-} L_h-\lambda \mathbb{V}^{+} L_h^T-\lambda \mathbb{W}+\lambda I$  and
$B_2= \lambda \mathbb{W}  - \lambda \mathbb{I}$.
By introducing the notation $U = [u^1, \cdots ,u^{N_t}]^T$,  one  obtains the following linear system
\begin{equation}\label{heatLFBoltzmann}
L U = F,
\end{equation}
where
\[ L =
\begin{bmatrix}
I  &            &           &            \\
-B & I     &           &            \\
        &\ddots      & \ddots    &    \\
        &            & -B   & I     \\
\end{bmatrix}
, \qquad
F =
\begin{bmatrix}
f^1 + Bu^0 \\
f^2  \\
\vdots\\
f^{N_t} \\
\end{bmatrix}.
\]
Note that $\mathcal{M}$ matrix is denoted
\[H = \begin{bmatrix} O & L \\  L^T  & O  \end{bmatrix}.\]

Let $\lambda$ be an eigenvalue of $H$, i.e.,
\[\mbox{det}(\lambda I-H) = \mbox{det}(\lambda^2I-LL^T) =:\mbox{det}(\mu I-LL^T),\]
where $\mu = \lambda^2$ is the eigenvalue of $LL^T$, or $\mu^{1/2}$ is the singular value of $L$. A direct calculation gives
\[LL^T =
\begin{bmatrix}
I  &    -B             &                   &            \\
-B  & I+B^2      & \ddots            &            \\
          &\ddots                   & \ddots            &   -B \\
          &                         & -B          & I+B^2    \\
\end{bmatrix}.
\]

We again establish the upper and lower bounds of the eigenvalues of $LL^T$ by using the Gershgorin circle theorem.

\begin{lemma}\label{lem:transpEig}
 The minimum and maximum eigenvalues of   $H$ satisfy
\[|\lambda|_{\min} =O(  \tau) , \qquad |\lambda|_{\max} =O(1). 
\]
\end{lemma}

\begin{proof}
Note 
\[
B^2 =
\begin{bmatrix}
B_1^2 &  B_1B_2+B_2B_1    &  B_2^2&     &                   &\\
& B_1^2&  B_1B_2+B_2B_1 &B_2^2 && \\
        &  &  \ddots                & \ddots   &\ddots         &  \\
&& &  B_1^2 &  B_1B_2+B_2B_1  &B_2^2    \\
&& & & B_1^2 &  B_1B_2+B_2B_1\\
&&&&&B_1^2
\end{bmatrix}. \\
\]
If $\lambda \le 1$, and also due to the condition \eqref{w-cond}, one has
$W^2=W$ and the sums of the each rows of $W$ is bounded by $1$, it is rather
easy to see that Gershgorin's theorm applied to $LL^T$ implies that
\[
\mu_{max} \le C
\]
where $C$ is a constant independent of $M, N_x$ and $N_t$. 

 To estimate the  maximum eigenvalues of $LL^T$, 
by definition, the spectral radius of $L$
\[\sigma_{\min}(L) = \frac{1}{\sigma_{\max}(L^{-1})}.\]
By simple algebra, one has
\begin{equation}
L^{-1} =
\begin{bmatrix}
I          &              &                 &          &            \\
B        & I      &                 &          &        \\
 B^2      &   B        & I       &                  &\\
   \vdots       &  \ddots      &  \ddots         &  \ddots  &   \\
B^{N_t-1} &   \cdots     &       &  B &  I      \\
\end{bmatrix}
=\begin{bmatrix}
I          &              &                 &          &            \\
                & I       &                 &          &        \\
                &              &  \ddots         &          &\\
                &              &                 &  \ddots  &   \\
                &              &                 &          &  I      \\
\end{bmatrix}+
\begin{bmatrix}
                 &              &                 &          &            \\
B        &              &                 &          &        \\
                 &  B      &                 &          &\\
                 &              &  \ddots         &          &   \\
                 &              &                 & B
\end{bmatrix} + \cdots, 
\end{equation}
which gives
\begin{align*}
\sigma_{\max}(L^{-1})
& = \|L^{-1}\|_2 \le \|I\|_2 + \|B\|_2 + \cdots +\|B^{N_t-1}\|_2 \\
& \le \|I\|_2 + \|B\|_2 + \|B\|_2^2 +\cdots +\|B\|_2^{N_t-1} 
\end{align*}
Note \[
BB^T =
\begin{bmatrix}
B_1B_1^T+B_2B_2^T & B_2B_1^T   &       &                   &\\
B_1B_2^T & B_1B_1^T+B_2B_2^T & B_2B_1^T  &  B_2B_1^T &&\\
          &  \ddots                & \ddots    &\ddots        &  \\
&& B_1B_2^T & B_1B_1^T+B_2B_2^T & B_2B_1^T  &    \\
&&& B_1B_2^T & B_1B_1^T+B_2B_2^T &  \\
 &                         & &   &  B_2B_1^T  \\
\end{bmatrix}, 
\]
Since $\|B\|_2$ is the square root of the maximum eigenvalue of $BB^T$, again by Gershgorin's theorem, under the condition that $\lambda\le 1$, one has
\[
\|B\|_2^2\le C,
\]
where $C$ is independent of $M, N_x$ and $N_t$. 
 Hence,
\begin{align*}
\sigma_{\max}(L^{-1}) \le C N_t=O(1/\tau)
\end{align*}
Thus, $\sigma_{\min}(L) =O( \tau)=O(1/N_x)$. The second equality is due to the CFL condition $\lambda=\tau/h=O(1)$.
\end{proof}

The above estimates of eigenvalues lead to $k=O(N_x)$.
The sparsity of  number $s =O(M)$ due to the integral collision term. 

For $x,v\in \mathbb{R}^d$ and $w\in \mathbb{R}^L$, similar to the multi-dimensional case for the linear heat equation, it is easy to get
\[
\kappa=O((d+L) N_x), \quad  s= O(LM^d).
  \]
  We omit the details of the proof.

\subsection{Linear advection equation}\label{AppendixB-3}

For the advection equation in phase space \eqref{phase-wave-1},
we use the forward Euler method in time, center difference in $p$ and upwind discretization in $x$. This equation contains third order derivatives, so the computational costs will be similar to those of the phase space heat equation.  The difference here is that we have the  boundary condition at $p_i=0$.

We first discuss the implementation of the boundary condition in \eqref{phase-wave-1}. Since it is a convection equation, a natural discretization is the upwind discretization:
\begin{equation}
\frac{u^{n}_{j1l} - u^n_{j0l}}{\tau} -
\frac{u^{n}_{j0,l+1} - u^n_{j0l}}{h}=0, \quad
n=1, \cdots, N_t, \quad j= 1,\cdots, N_x-1; \quad l=0,\cdots, N_x-1
\end{equation}
or
\begin{equation}\label{BC-UW}
    (1-\tau/h)u^n_{j0l}+u^n_{j0,l+1}=u^{n}_{j1l}
\end{equation}
This involves the coupling of two grids in $q$, which will make the  the matrix structure that defines the QLSP to be solved more complicated. Thus computing the condition number of the matrix to be inverted becomes difficult. To simplify the analysis,
we first choose $\tau=h$ in \eqref{BC-UW}. This corresponds to the method of characteristics for the boundary condition. Then \eqref{BC-UW} becomes
\begin{equation}\label{BC-UW-1-adv}
    u^n_{j0,l+1}=u^{n}_{j1l}
\end{equation}
To avoid the coupling between different grids in $q$ we further use
\[
u^n_{j0,l+1}=u^n_{j0,l}+O(h) 
\]
Now the boundary condition \label{BC-UW-1} becomes, after ignoring the $O(h)$ term,
\[
u^n_{j0,l+1}=u^n_{j0,l}
\]
which is just the numerical implementation of the Neumann boundary condition. At $p_i=1$ the Dirichet boundary condition is used. Since the largest discrete eigenvalues of the discrete Laplacian with  mixed Neumann-Dirichlet boundary condition is the same order of that with the  Dirichlet boundary condition \cite{kuo2022boundary}, a similar analysis as in section \ref{AppendixB-1} gives
\[
  \kappa= O(LdN_x^3), \qquad s=O(Ld)\,.
\]

\subsection{The Schr\"odinger  equation}\label{AppendixB-4}

For the Schr\"odinger equation in phase space \eqref{phase-Schr-2}, the equations have  second order derivatives in both $x$ and $s$, but the highest order derivative remains second order, so the equation is similar to that of the original Schr\"odinger equation \eqref{Schr} except the total dimension increases from $d$ to $d+L$, and one has the  boundary condition at $w_i=0$. We use forward Euler method in time and  center difference in $x$ and $w$. We use the same apporximation of the boundary condition as the case of the convection equation, 
then  a similar analysis as in the case of heat equation \cite{JinLiuYu} gives:
\[\kappa=O((d+L)N_x^2),\, \, s=O(d+L).
\]

 \section{Numerical discretisations and classical cost}
 
 In this section we always assume that the spatial derivatives are approximated by $r^{\text{th}}$ order method and explicit time discretization is used (although in the previous section our analysis was done for given $r$'s.
 
 \subsection{Linear heat equation} \label{app:classicalheat}
 
 For the original linear heat equation Eq.~\eqref{heat}, the classical cost is $\mathcal{C}_{can}=O(M\bar{N}^d\bar{N}_td)$. To compute how $\bar{N}$ and $\bar{N}_t$ depend on $d$ and $\epsilon$, we require the following error and CFL conditions
 \begin{align} \label{app:heat1}
     \begin{cases}
     Error: \qquad O(d(\Delta x)^{r})=O(d/\bar{N}^{r})=O(\epsilon) \\
     CFL: \qquad O(1/\bar{N}_t)=\Delta t=O((\Delta x)^2/d)=O(1/(\bar{N}^2 d))
     \end{cases}
 \end{align}
 for $r^{\text{th}}$-order methods, where $c\geq 1$. Then Eq.~\eqref{app:heat1} implies 
\begin{align}
 \bar{N}=O((d/\epsilon)^{1/r}), \qquad \bar{N}_t=O(d(d/\epsilon)^{2/r}).
\end{align}
Therefore
\begin{align}
    \mathcal{C}_{can}=O(M\bar{N}^d\bar{N}_td)=O(Md^2(d/\epsilon)^{(d+2)/r}).
\end{align}

For the modified linear heat equation in Eq.~\eqref{phase-heat-11}, the classical cost is $\mathcal{C}_{mod}=O(LN^dN_p^LN_td)$. We define $N=N_p$.  To compute how $N=N_p$ and $N_t$ depend on $d$ and $\epsilon$, we require the following error and CFL conditions
 \begin{align} \label{app:heat2}
     \begin{cases}
     Error: \qquad O(Ld(\Delta x)^r)=O(Ld/N^r)=O(\epsilon) \\
     CFL: \qquad O(1/N_t)=\Delta t=O((\Delta x)^3/(L(d+1)))=O(1/(LdN^3))
     \end{cases}
 \end{align}
 for $r^{\text{th}}$-order methods. Eq.~\eqref{app:heat2} then implies
 \begin{align}
     N=O((Ld/\epsilon)^{1/r}), \qquad N_t=O((Ld)^{1+3/r}/\epsilon^{3/r}).
 \end{align}
 Therefore
 \begin{align}
     \mathcal{C}_{mod}=O(LN^dN_p^LN_td)=O((Ld)^{2+(d+L+3)/r}/\epsilon^{(d+L+3)/r}).
 \end{align}
 We want to compare quantum costs to the minimal classical cost. Here, the minimal classical cost $\mathcal{C}=\mathcal{C}_{can}<\mathcal{C}_{mod}$ when $M<M_{heat} \equiv O(L^{2+(d+L+3)/r}(d/\epsilon)^{(L+1)/3})$. If $M$ is very large and $d, L$ are relatively small so we can have the condition $M>M_{heat}$, then $\mathcal{C}=\mathcal{C}_{mod}<\mathcal{C}_{can}$.
 
 \subsection{Linear Boltzmann equation} \label{app:classicaltransp}
 
  For the original linear Boltzmann equation Eq.~\eqref{Transp}, the classical cost is $\mathcal{C}_{can}=O(M\bar{N}^d\bar{N}_v^d\bar{N}_t)$.  To compute how $\bar{N}$ and $\bar{N}_t$ depend on $d$ and $\epsilon$, we require the following error and CFL conditions
   \begin{align} \label{app:transport1}
     \begin{cases}
     Error: \qquad O(d\Delta x^r+d(\Delta v)^{r})=O(d/\bar{N}^r+d/\bar{N}_v^
     {r})=O(\epsilon) \\
     CFL: \qquad O(1/\bar{N}_t)=\Delta t=O(\Delta x/d)=O(1/(\bar{N}d))
     \end{cases}
 \end{align}
  We define $\bar{N}=\bar{N}_p$.
 Then Eq.~\eqref{app:transport1} implies
 \begin{align}
     \bar{N}=O((d/\epsilon)^{1/r}), \qquad  \bar{N}_t=O(d(d/\epsilon)^{1/r}).
 \end{align}
 Therefore
 \begin{align}
     \mathcal{C}_{can}=O(M\bar{N}^d\bar{N}_v^d\bar{N}_t)=O(Md(d/\epsilon)^{(2d+1)/r}).
 \end{align}
 For the modified linear Boltzmann equation in Eq.~\eqref{phase-Transp-11}, the classical cost is $\mathcal{C}_{mod}=O(N^dN_p^LN_v^dN_t)$. We define $N=N_v=N_p$. To compute how $N$ and $N_t$ depend on $d$ and $\epsilon$, we require the following error and CFL conditions
 \begin{align} \label{app:transport2}
     \begin{cases}
     Error: \qquad O(d (\Delta x)^r+L(d(\Delta v)^{r}+(\Delta w)^r))=O((d+L+Ld)/N^r)=O(\epsilon) \\
     CFL: \qquad O(1/N_t)=\Delta t=\min(\Delta x/d, \Delta w/L)=\min(1/L, 1/d)/N
     \end{cases}
 \end{align}
Then Eq.~\eqref{app:transport2} implies
 \begin{align}
   N=O(((d+L+Ld)/\epsilon)^{1/r})=O((Ld/\epsilon)^{1/r}), \qquad  N_t=O(\max(L, d)(Ld/\epsilon)^{1/r})
 \end{align}
 Therefore
 \begin{align}
 \mathcal{C}_{mod}=O(N^dN_p^LN_v^dN_t)=O((Ld/\epsilon)^{(2d+L+1)/r}\max(L, d)).
 \end{align}
 We want to compare the quantum cost to the minimal classical cost. Here the minimal classical cost is $\mathcal{C}=\mathcal{C}_{can}<\mathcal{C}_{mod}$ when $M<M_{Boltz} \equiv O(L^{(2d+L+1)/r}\max(L,d)(d/\epsilon)^{L/r}/d)$. If $M>M_{Boltz}$, then $\mathcal{C}=\mathcal{C}_{mod}<\mathcal{C}_{can}$.

 \subsection{Linear advection equation} \label{app:classicalwave}
 For the original linear heat equation Eq.~\eqref{wave-eqn}, the classical cost is $\mathcal{C}_{can}=O(M\bar{N}^d\bar{N}_td)$.  To compute how $\bar{N}$ and $\bar{N}_t$ depend on $d$ and $\epsilon$, we require the following error and CFL conditions
 \begin{align} \label{app:wave1}
     \begin{cases}
     Error: \qquad O(d (\Delta x)^r)=O(d/\bar{N}^r)=O(\epsilon)  \\
     CFL: \qquad O(1/\bar{N}_t)=\Delta t=O(\Delta x/d)=O(1/(\bar{N}d))
     \end{cases}
 \end{align}
Eq.~\eqref{app:wave1} implies 
\begin{align}
 \bar{N}=O((d/\epsilon)^{1/r}), \qquad \bar{N}_t=O(d(d/\epsilon)^{1/r}).
\end{align}
Therefore
\begin{align}
    \mathcal{C}_{can}=O(M\bar{N}^d\bar{N}_td)=O(Md^{2}(d/\epsilon)^{(d+1)/r}).
\end{align}
For the modified linear advection equation in Eq.~\eqref{phase-wave-1}, the classical cost is $\mathcal{C}_{mod}=O(LN^dN_p^LN_q^LN_td)$. We define $N=N_p=N_q$.  To compute how $N=N$ and $N_t$ depend on $d$ and $\epsilon$, we require the following error and CFL conditions
 \begin{align} \label{app:wave2}
     \begin{cases}
     Error:  \qquad O(Ld(\Delta x)^r)=O(Ld/N^r)=O(\epsilon)\\
     CFL: \qquad O(1/N_t)=\Delta t=O((\Delta x)^3/(Ld))=O(1/(LN^3d))
     \end{cases}
 \end{align}
 Eq.~\eqref{app:wave2} implies
 \begin{align}
     N=O((Ld/\epsilon)^{1/r}), \qquad N_t=O(Ld(Ld/\epsilon)^{3/r}).
 \end{align}
 Therefore
 \begin{align}
     \mathcal{C}_{mod}=O(LN^dN_p^LN_q^LN_td)=O((Ld)^2(Ld/\epsilon)^{(d+2L+3)/r}).
 \end{align}
  We want to compare quantum costs to the minimal classical cost. Here, the minimal classical cost $\mathcal{C}=\mathcal{C}_{can}<\mathcal{C}_{mod}$ when $M<M_{adv} \equiv O(L^{(d+2L+3)/r+2}(d/\epsilon)^{(2L+2)/r})$. If $M$ is very large and $d, L$ are relatively small so we can have the condition $M>M_{adv}$, then $\mathcal{C}=\mathcal{C}_{mod}<\mathcal{C}_{can}$.
 
 \subsection{The Schr\"odinger equation} \label{app:classicalSchr}
 
  For the original linear heat equation Eq.~\eqref{Schr}, the classical cost is $\mathcal{C}_{can}=O(M\bar{N}^d\bar{N}_td)$. To compute how $\bar{N}$ and $\bar{N}_t$ depend on $d$ and $\epsilon$, we require the following error and CFL conditions
 \begin{align} \label{app:Schr1}
     \begin{cases}
     Error: \qquad   O(d(\Delta x)^{r})=O(d/\bar{N}^{r})=O(\epsilon)\\
     CFL: \qquad O(1/\bar{N}_t)=\Delta t=O((\Delta x)^2/d)=O(1/(N^2d))
     \end{cases}
 \end{align}
Eq.~\eqref{app:Schr1} implies 
\begin{align}
 \bar{N}=O((d/\epsilon)^{1/r}), \qquad \bar{N}_t=O(d(d/\epsilon)^{2/r}).
\end{align}
Therefore
\begin{align}
    \mathcal{C}_{can}=O(M\bar{N}^d\bar{N}_td)=O(Md^2(d/\epsilon)^{(d+2)/r})
\end{align}
For the modified Schr\"odinger equation in Eq.~\eqref{phase-Schr-2}, the classical cost is $\mathcal{C}_{mod}=O(N^dN_p^LN_q^LN_t(d+L))$. We define $N=N_p=N_q$.  To compute how $N$ and $N_t$ depend on $d$ and $\epsilon$, we require the following error and CFL conditions
 \begin{align} \label{app:Schr2}
     \begin{cases}
     Error:  \qquad O(d(\Delta x)^{r}+L(\Delta w)^{r})=O((d+L)/N^{r})=O(\epsilon)\\
     CFL: \qquad O(1/N_t)=\Delta t=(\Delta x)^2/(d+L)=1/(N^2(d+L))
     \end{cases}
 \end{align}
 Eq.~\eqref{app:Schr2} implies
 \begin{align}
     N=O(((d+L)/\epsilon)^{1/r}), \qquad N_t=O((d+L)((d+L)/\epsilon)^{2/r}).
 \end{align}
 Therefore
 \begin{align}
     \mathcal{C}_{mod}=O(N^dN^{2L}N_t(d+L))=O((d+L)^2((d+L)/\epsilon)^{(d+2L+2)/r}).
 \end{align}
  We want to compare quantum costs to the minimal classical cost. Here, the minimal classical cost $\mathcal{C}=\mathcal{C}_{can}<\mathcal{C}_{mod}$ when $M<M_{Schr} \equiv O((d+L)^{(d+2L+2)/r+2}/(d^{(d+2)/r+2}\epsilon^{2L/r}))$. If $M$ is very large and $d, L$ are relatively small so we can have the condition $M>M_{Schr}$, then $\mathcal{C}=\mathcal{C}_{mod}<\mathcal{C}_{can}$.

\section{Quantum cost of solving original PDEs} \label{app:qoriginal}
To solve for ensemble averaged observables of the original PDEs (not in phase space) with $M$ initial conditions, each PDE needs to be solved $M$ times separately. This means for each initial condition, a different $\mathcal{M}$, hence a different sparse-access is required. There may be cases where this can be either costly to obtain or is not available, hence solving the phase space equations is preferred.

Suppose we assume that the $M$ different sparse access to $\mathcal{M}$ are available and can be accessed without extra cost. The quantum cost to compute the observables for the linear heat, linear Boltzmann and linear advection equations can be obtained from Lemma~\ref{lem:bslep2}. Here the sparsity and condition nubmers of the corresponding $\mathcal{M}$ are found in Appendix~\ref{app:numerics}.

For the linear heat equation, the sparsity $s$ and conditions number $\kappa$ of $\mathcal{M}$ are $O(d)$ and $O(N^2)$ respectively, where $N=O(d/\epsilon)^{1/r}$ for $r^{\text{th}}$-order methods. Lemma~\ref{lem:bslep2} can be used directly to arrive at $\mathcal{Q}_{orig}=\tilde{\mathcal{O}}(Ms\kappa^3\mathcal{N}_y^2/\epsilon)$ to compute the observable crorresponding to $M$ initial conditions to precision $\epsilon$, where the normalisation $\mathcal{N}_y$ is the normalisation corresponding to the initial condition of the original PDE, which is a factor of order $O(N_p^L)$ less than the normalisation of the initial state for the phase space equation. This means that $\mathcal{N}^2_y$ term is comparable in size to the normalisation $n^2_{W_0}$ corresponding to the phase space equation. Then comparing $\mathcal{Q}_{orig}$ with $\mathcal{Q}$ from Theorem~\ref{thm:linearheat}, we have $\mathcal{Q}<\mathcal{Q}_{orig}$ when $M>\tilde{\mathcal{O}}(L^{4+9/r}d^3(d/\epsilon)^{3/r})$. 

The analysis is the same for the linear Boltzmann equation ($s=O(N_v^d)$ and $\kappa=O(dN)$) and the linear advection equation ($s=O(d)$ and $\kappa=O(dN)$) is the same, with $N=N_v=O(d/\epsilon)^{1/r}$. For the Schr\"odinger equation, we can still use Lemma~\ref{lem:bslep2}, but potentially better scaling is possible when one uses either Lemma~\ref{lem:bslep3} or quantum simulation methods. For consistency in also using a matrix inversion method, we use Lemma~\ref{lem:bslep3} where $s=O(d)$ and $\kappa=O(N^2)$ for $N=(d/\epsilon)^{1/r}$. Quantum simulation methods will differ only by small polynomial factors, which will not significantly affect the final comparison to $\mathcal{Q}$. We omit details of the proof since it is straightforward.  

\bibliography{QUQ}

\begin{thebibliography}{37}%
\makeatletter
\providecommand \@ifxundefined [1]{%
 \@ifx{#1\undefined}
}%
\providecommand \@ifnum [1]{%
 \ifnum #1\expandafter \@firstoftwo
 \else \expandafter \@secondoftwo
 \fi
}%
\providecommand \@ifx [1]{%
 \ifx #1\expandafter \@firstoftwo
 \else \expandafter \@secondoftwo
 \fi
}%
\providecommand \natexlab [1]{#1}%
\providecommand \enquote  [1]{``#1''}%
\providecommand \bibnamefont  [1]{#1}%
\providecommand \bibfnamefont [1]{#1}%
\providecommand \citenamefont [1]{#1}%
\providecommand \href@noop [0]{\@secondoftwo}%
\providecommand \href [0]{\begingroup \@sanitize@url \@href}%
\providecommand \@href[1]{\@@startlink{#1}\@@href}%
\providecommand \@@href[1]{\endgroup#1\@@endlink}%
\providecommand \@sanitize@url [0]{\catcode `\\12\catcode `\$12\catcode
  `\&12\catcode `\#12\catcode `\^12\catcode `\_12\catcode `\%12\relax}%
\providecommand \@@startlink[1]{}%
\providecommand \@@endlink[0]{}%
\providecommand \url  [0]{\begingroup\@sanitize@url \@url }%
\providecommand \@url [1]{\endgroup\@href {#1}{\urlprefix }}%
\providecommand \urlprefix  [0]{URL }%
\providecommand \Eprint [0]{\href }%
\providecommand \doibase [0]{https://doi.org/}%
\providecommand \selectlanguage [0]{\@gobble}%
\providecommand \bibinfo  [0]{\@secondoftwo}%
\providecommand \bibfield  [0]{\@secondoftwo}%
\providecommand \translation [1]{[#1]}%
\providecommand \BibitemOpen [0]{}%
\providecommand \bibitemStop [0]{}%
\providecommand \bibitemNoStop [0]{.\EOS\space}%
\providecommand \EOS [0]{\spacefactor3000\relax}%
\providecommand \BibitemShut  [1]{\csname bibitem#1\endcsname}%
\let\auto@bib@innerbib\@empty
\bibitem [{\citenamefont {Smith}(2013)}]{smith2013}%
  \BibitemOpen
  \bibfield  {author} {\bibinfo {author} {\bibfnamefont {R.~C.}\ \bibnamefont
  {Smith}},\ }\href@noop {} {\emph {\bibinfo {title} {Uncertainty
  quantification: theory, implementation, and applications}}},\ Vol.~\bibinfo
  {volume} {12}\ (\bibinfo  {publisher} {Siam},\ \bibinfo {year}
  {2013})\BibitemShut {NoStop}%
\bibitem [{\citenamefont {Ghanem}\ \emph {et~al.}(2017)\citenamefont {Ghanem},
  \citenamefont {Higdon}, \citenamefont {Owhadi} \emph {et~al.}}]{ghanem2017}%
  \BibitemOpen
  \bibfield  {author} {\bibinfo {author} {\bibfnamefont {R.}~\bibnamefont
  {Ghanem}}, \bibinfo {author} {\bibfnamefont {D.}~\bibnamefont {Higdon}},
  \bibinfo {author} {\bibfnamefont {H.}~\bibnamefont {Owhadi}}, \emph
  {et~al.},\ }\href@noop {} {\emph {\bibinfo {title} {Handbook of uncertainty
  quantification}}},\ Vol.~\bibinfo {volume} {6}\ (\bibinfo  {publisher}
  {Springer},\ \bibinfo {year} {2017})\BibitemShut {NoStop}%
\bibitem [{\citenamefont {Xiu}(2010)}]{xiu2010}%
  \BibitemOpen
  \bibfield  {author} {\bibinfo {author} {\bibfnamefont {D.}~\bibnamefont
  {Xiu}},\ }\bibfield  {title} {\bibinfo {title} {Numerical methods for
  stochastic computations},\ }in\ \href@noop {} {\emph {\bibinfo {booktitle}
  {Numerical Methods for Stochastic Computations}}}\ (\bibinfo  {publisher}
  {Princeton university press},\ \bibinfo {year} {2010})\BibitemShut {NoStop}%
\bibitem [{\citenamefont {Gunzburger}\ \emph {et~al.}(2014)\citenamefont
  {Gunzburger}, \citenamefont {Webster},\ and\ \citenamefont
  {Zhang}}]{GWZ-Acta}%
  \BibitemOpen
  \bibfield  {author} {\bibinfo {author} {\bibfnamefont {M.~D.}\ \bibnamefont
  {Gunzburger}}, \bibinfo {author} {\bibfnamefont {C.~G.}\ \bibnamefont
  {Webster}},\ and\ \bibinfo {author} {\bibfnamefont {G.}~\bibnamefont
  {Zhang}},\ }\bibfield  {title} {\bibinfo {title} {Stochastic finite element
  methods for partial differential equations with random input data},\
  }\href@noop {} {\bibfield  {journal} {\bibinfo  {journal} {Acta Numerica}\
  }\textbf {\bibinfo {volume} {23}},\ \bibinfo {pages} {521} (\bibinfo {year}
  {2014})}\BibitemShut {NoStop}%
\bibitem [{\citenamefont {Ghanem}\ and\ \citenamefont {Spanos}(2003)}]{GS03}%
  \BibitemOpen
  \bibfield  {author} {\bibinfo {author} {\bibfnamefont {R.~G.}\ \bibnamefont
  {Ghanem}}\ and\ \bibinfo {author} {\bibfnamefont {P.~D.}\ \bibnamefont
  {Spanos}},\ }\href@noop {} {\emph {\bibinfo {title} {Stochastic finite
  elements: a spectral approach}}}\ (\bibinfo  {publisher} {Courier
  Corporation},\ \bibinfo {year} {2003})\BibitemShut {NoStop}%
\bibitem [{\citenamefont {Bungartz}\ and\ \citenamefont
  {Griebel}(2004)}]{bungartz2004}%
  \BibitemOpen
  \bibfield  {author} {\bibinfo {author} {\bibfnamefont {H.-J.}\ \bibnamefont
  {Bungartz}}\ and\ \bibinfo {author} {\bibfnamefont {M.}~\bibnamefont
  {Griebel}},\ }\bibfield  {title} {\bibinfo {title} {Sparse grids},\
  }\href@noop {} {\bibfield  {journal} {\bibinfo  {journal} {Acta numerica}\
  }\textbf {\bibinfo {volume} {13}},\ \bibinfo {pages} {147} (\bibinfo {year}
  {2004})}\BibitemShut {NoStop}%
\bibitem [{\citenamefont {Cohen}\ and\ \citenamefont
  {DeVore}(2015)}]{cohen2015}%
  \BibitemOpen
  \bibfield  {author} {\bibinfo {author} {\bibfnamefont {A.}~\bibnamefont
  {Cohen}}\ and\ \bibinfo {author} {\bibfnamefont {R.}~\bibnamefont {DeVore}},\
  }\bibfield  {title} {\bibinfo {title} {Approximation of high-dimensional
  parametric pdes},\ }\href@noop {} {\bibfield  {journal} {\bibinfo  {journal}
  {Acta Numerica}\ }\textbf {\bibinfo {volume} {24}},\ \bibinfo {pages} {1}
  (\bibinfo {year} {2015})}\BibitemShut {NoStop}%
\bibitem [{\citenamefont {Harrow}\ \emph {et~al.}(2009)\citenamefont {Harrow},
  \citenamefont {Hassidim},\ and\ \citenamefont {Lloyd}}]{HHL-2009}%
  \BibitemOpen
  \bibfield  {author} {\bibinfo {author} {\bibfnamefont {A.~W.}\ \bibnamefont
  {Harrow}}, \bibinfo {author} {\bibfnamefont {A.}~\bibnamefont {Hassidim}},\
  and\ \bibinfo {author} {\bibfnamefont {S.}~\bibnamefont {Lloyd}},\ }\bibfield
   {title} {\bibinfo {title} {Quantum algorithm for linear systems of
  equations},\ }\href@noop {} {\bibfield  {journal} {\bibinfo  {journal} {Phys.
  Rev. Lett.}\ }\textbf {\bibinfo {volume} {103}},\ \bibinfo {pages} {150502, 4
  pp.} (\bibinfo {year} {2009})}\BibitemShut {NoStop}%
\bibitem [{\citenamefont {Childs}\ \emph {et~al.}(2017)\citenamefont {Childs},
  \citenamefont {Kothari},\ and\ \citenamefont {Somma}}]{Childs-2017}%
  \BibitemOpen
  \bibfield  {author} {\bibinfo {author} {\bibfnamefont {A.~M.}\ \bibnamefont
  {Childs}}, \bibinfo {author} {\bibfnamefont {R.}~\bibnamefont {Kothari}},\
  and\ \bibinfo {author} {\bibfnamefont {R.~D.}\ \bibnamefont {Somma}},\
  }\bibfield  {title} {\bibinfo {title} {Quantum algorithm for systems of
  linear equations with exponentially improved dependence on precision},\
  }\href@noop {} {\bibfield  {journal} {\bibinfo  {journal} {SIAM J. Comput.}\
  }\textbf {\bibinfo {volume} {46}},\ \bibinfo {pages} {1920} (\bibinfo {year}
  {2017})}\BibitemShut {NoStop}%
\bibitem [{\citenamefont {Gily{\'e}n}\ \emph {et~al.}(2019)\citenamefont
  {Gily{\'e}n}, \citenamefont {Su}, \citenamefont {Low},\ and\ \citenamefont
  {Wiebe}}]{qsvd}%
  \BibitemOpen
  \bibfield  {author} {\bibinfo {author} {\bibfnamefont {A.}~\bibnamefont
  {Gily{\'e}n}}, \bibinfo {author} {\bibfnamefont {Y.}~\bibnamefont {Su}},
  \bibinfo {author} {\bibfnamefont {G.~H.}\ \bibnamefont {Low}},\ and\ \bibinfo
  {author} {\bibfnamefont {N.}~\bibnamefont {Wiebe}},\ }\bibfield  {title}
  {\bibinfo {title} {Quantum singular value transformation and beyond:
  exponential improvements for quantum matrix arithmetics},\ }in\ \href@noop {}
  {\emph {\bibinfo {booktitle} {Proceedings of the 51st Annual ACM SIGACT
  Symposium on Theory of Computing}}}\ (\bibinfo {year} {2019})\ pp.\ \bibinfo
  {pages} {193--204}\BibitemShut {NoStop}%
\bibitem [{\citenamefont {Berry}(2014)}]{Berry-2014}%
  \BibitemOpen
  \bibfield  {author} {\bibinfo {author} {\bibfnamefont {D.~W.}\ \bibnamefont
  {Berry}},\ }\bibfield  {title} {\bibinfo {title} {High-order quantum
  algorithm for solving linear differential equations},\ }\href@noop {}
  {\bibfield  {journal} {\bibinfo  {journal} {J. Phys. A: Math. Theor.}\
  }\textbf {\bibinfo {volume} {47}},\ \bibinfo {pages} {105301, 17 pp.}
  (\bibinfo {year} {2014})}\BibitemShut {NoStop}%
\bibitem [{\citenamefont {Joseph}(2020)}]{joseph2020}%
  \BibitemOpen
  \bibfield  {author} {\bibinfo {author} {\bibfnamefont {I.}~\bibnamefont
  {Joseph}},\ }\bibfield  {title} {\bibinfo {title} {Koopman--von neumann
  approach to quantum simulation of nonlinear classical dynamics},\ }\href@noop
  {} {\bibfield  {journal} {\bibinfo  {journal} {Physical Review Research}\
  }\textbf {\bibinfo {volume} {2}},\ \bibinfo {pages} {043102} (\bibinfo {year}
  {2020})}\BibitemShut {NoStop}%
\bibitem [{\citenamefont {Dodin}\ and\ \citenamefont
  {Startsev}(2021)}]{dodin2021}%
  \BibitemOpen
  \bibfield  {author} {\bibinfo {author} {\bibfnamefont {I.~Y.}\ \bibnamefont
  {Dodin}}\ and\ \bibinfo {author} {\bibfnamefont {E.~A.}\ \bibnamefont
  {Startsev}},\ }\bibfield  {title} {\bibinfo {title} {On applications of
  quantum computing to plasma simulations},\ }\href@noop {} {\bibfield
  {journal} {\bibinfo  {journal} {Physics of Plasmas}\ }\textbf {\bibinfo
  {volume} {28}},\ \bibinfo {pages} {092101} (\bibinfo {year}
  {2021})}\BibitemShut {NoStop}%
\bibitem [{\citenamefont {Lloyd}\ \emph {et~al.}(2020)\citenamefont {Lloyd},
  \citenamefont {De~Palma}, \citenamefont {Gokler}, \citenamefont {Kiani},
  \citenamefont {Liu}, \citenamefont {Marvian}, \citenamefont {Tennie},\ and\
  \citenamefont {Palmer}}]{lloyd2020quantum}%
  \BibitemOpen
  \bibfield  {author} {\bibinfo {author} {\bibfnamefont {S.}~\bibnamefont
  {Lloyd}}, \bibinfo {author} {\bibfnamefont {G.}~\bibnamefont {De~Palma}},
  \bibinfo {author} {\bibfnamefont {C.}~\bibnamefont {Gokler}}, \bibinfo
  {author} {\bibfnamefont {B.}~\bibnamefont {Kiani}}, \bibinfo {author}
  {\bibfnamefont {Z.-W.}\ \bibnamefont {Liu}}, \bibinfo {author} {\bibfnamefont
  {M.}~\bibnamefont {Marvian}}, \bibinfo {author} {\bibfnamefont
  {F.}~\bibnamefont {Tennie}},\ and\ \bibinfo {author} {\bibfnamefont
  {T.}~\bibnamefont {Palmer}},\ }\bibfield  {title} {\bibinfo {title} {Quantum
  algorithm for nonlinear differential equations},\ }\href@noop {} {\bibfield
  {journal} {\bibinfo  {journal} {arXiv preprint arXiv:2011.06571}\ } (\bibinfo
  {year} {2020})}\BibitemShut {NoStop}%
\bibitem [{\citenamefont {Childs}\ \emph {et~al.}(2021)\citenamefont {Childs},
  \citenamefont {Liu},\ and\ \citenamefont {Ostrander}}]{Childs-2021}%
  \BibitemOpen
  \bibfield  {author} {\bibinfo {author} {\bibfnamefont {A.~M.}\ \bibnamefont
  {Childs}}, \bibinfo {author} {\bibfnamefont {J.~P.}\ \bibnamefont {Liu}},\
  and\ \bibinfo {author} {\bibfnamefont {A.}~\bibnamefont {Ostrander}},\
  }\bibfield  {title} {\bibinfo {title} {High-precision quantum algorithms for
  partial differential equations},\ }\href@noop {} {\bibfield  {journal}
  {\bibinfo  {journal} {Quantum}\ }\textbf {\bibinfo {volume} {5}},\ \bibinfo
  {pages} {574} (\bibinfo {year} {2021})}\BibitemShut {NoStop}%
\bibitem [{\citenamefont {Liu}\ \emph {et~al.}(2021)\citenamefont {Liu},
  \citenamefont {Kolden}, \citenamefont {Krovi}, \citenamefont {Loureiro},
  \citenamefont {Trivisa},\ and\ \citenamefont {Childs}}]{Liu2021nonlinear}%
  \BibitemOpen
  \bibfield  {author} {\bibinfo {author} {\bibfnamefont {J.}~\bibnamefont
  {Liu}}, \bibinfo {author} {\bibfnamefont {H.~{\O}.}\ \bibnamefont {Kolden}},
  \bibinfo {author} {\bibfnamefont {H.~K.}\ \bibnamefont {Krovi}}, \bibinfo
  {author} {\bibfnamefont {N.~F.}\ \bibnamefont {Loureiro}}, \bibinfo {author}
  {\bibfnamefont {K.}~\bibnamefont {Trivisa}},\ and\ \bibinfo {author}
  {\bibfnamefont {A.~M.}\ \bibnamefont {Childs}},\ }\bibfield  {title}
  {\bibinfo {title} {Efficient quantum algorithm for dissipative nonlinear
  differential equations},\ }\href@noop {} {\bibfield  {journal} {\bibinfo
  {journal} {Proc. Natl. Acad. Sci. U. S. A.}\ }\textbf {\bibinfo {volume}
  {118}} (\bibinfo {year} {2021})}\BibitemShut {NoStop}%
\bibitem [{\citenamefont {Jin}\ \emph {et~al.}(2022{\natexlab{a}})\citenamefont
  {Jin}, \citenamefont {Liu},\ and\ \citenamefont {Yu}}]{jin2022time}%
  \BibitemOpen
  \bibfield  {author} {\bibinfo {author} {\bibfnamefont {S.}~\bibnamefont
  {Jin}}, \bibinfo {author} {\bibfnamefont {N.}~\bibnamefont {Liu}},\ and\
  \bibinfo {author} {\bibfnamefont {Y.}~\bibnamefont {Yu}},\ }\bibfield
  {title} {\bibinfo {title} {Time complexity analysis of quantum difference
  methods for linear high dimensional and multiscale partial differential
  equations},\ }\href@noop {} {\bibfield  {journal} {\bibinfo  {journal} {arXiv
  preprint arXiv:2202.04537}\ } (\bibinfo {year}
  {2022}{\natexlab{a}})}\BibitemShut {NoStop}%
\bibitem [{\citenamefont {Jin}\ and\ \citenamefont {Liu}(2022)}]{jin-liu-2022}%
  \BibitemOpen
  \bibfield  {author} {\bibinfo {author} {\bibfnamefont {S.}~\bibnamefont
  {Jin}}\ and\ \bibinfo {author} {\bibfnamefont {N.}~\bibnamefont {Liu}},\
  }\bibfield  {title} {\bibinfo {title} {Quantum algorithms for computing
  observables of nonlinear partial differential equations},\ }\href@noop {}
  {\bibfield  {journal} {\bibinfo  {journal} {arXiv preprint arXiv:2202.07834}\
  } (\bibinfo {year} {2022})}\BibitemShut {NoStop}%
\bibitem [{\citenamefont {Bernard}\ \emph {et~al.}(2010)\citenamefont
  {Bernard}, \citenamefont {Golse},\ and\ \citenamefont {Salvarani}}]{BGS2010}%
  \BibitemOpen
  \bibfield  {author} {\bibinfo {author} {\bibfnamefont {E.}~\bibnamefont
  {Bernard}}, \bibinfo {author} {\bibfnamefont {F.}~\bibnamefont {Golse}},\
  and\ \bibinfo {author} {\bibfnamefont {F.}~\bibnamefont {Salvarani}},\
  }\bibfield  {title} {\bibinfo {title} {Homogenization of transport problems
  and semigroups},\ }\href@noop {} {\bibfield  {journal} {\bibinfo  {journal}
  {Mathematical methods in the applied sciences}\ }\textbf {\bibinfo {volume}
  {33}},\ \bibinfo {pages} {1228} (\bibinfo {year} {2010})}\BibitemShut
  {NoStop}%
\bibitem [{\citenamefont {Ambainis}(2012)}]{Ambainis-2012}%
  \BibitemOpen
  \bibfield  {author} {\bibinfo {author} {\bibfnamefont {A.}~\bibnamefont
  {Ambainis}},\ }\bibfield  {title} {\bibinfo {title} {Variable time amplitude
  amplification and quantum algorithms for linear algebra problems},\
  }\href@noop {} {\bibfield  {journal} {\bibinfo  {journal} {LIPIcs. Leibniz
  Int. Proc. Inform.}\ }\textbf {\bibinfo {volume} {14}},\ \bibinfo {pages}
  {636} (\bibinfo {year} {2012})}\BibitemShut {NoStop}%
\bibitem [{\citenamefont {Lin}\ and\ \citenamefont
  {Tong}(2020)}]{Lin2020optimalpolynomial}%
  \BibitemOpen
  \bibfield  {author} {\bibinfo {author} {\bibfnamefont {L.}~\bibnamefont
  {Lin}}\ and\ \bibinfo {author} {\bibfnamefont {Y.}~\bibnamefont {Tong}},\
  }\bibfield  {title} {\bibinfo {title} {Optimal polynomial based quantum
  eigenstate filtering with application to solving quantum linear systems},\
  }\href {https://doi.org/10.22331/q-2020-11-11-361} {\bibfield  {journal}
  {\bibinfo  {journal} {{Quantum}}\ }\textbf {\bibinfo {volume} {4}},\ \bibinfo
  {pages} {361} (\bibinfo {year} {2020})}\BibitemShut {NoStop}%
\bibitem [{\citenamefont {Costa}\ \emph {et~al.}(2021)\citenamefont {Costa},
  \citenamefont {An}, \citenamefont {Sanders}, \citenamefont {Su},
  \citenamefont {Babbush},\ and\ \citenamefont {Berry}}]{costa2021optimal}%
  \BibitemOpen
  \bibfield  {author} {\bibinfo {author} {\bibfnamefont {P.}~\bibnamefont
  {Costa}}, \bibinfo {author} {\bibfnamefont {D.}~\bibnamefont {An}}, \bibinfo
  {author} {\bibfnamefont {Y.~R.}\ \bibnamefont {Sanders}}, \bibinfo {author}
  {\bibfnamefont {Y.}~\bibnamefont {Su}}, \bibinfo {author} {\bibfnamefont
  {R.}~\bibnamefont {Babbush}},\ and\ \bibinfo {author} {\bibfnamefont {D.~W.}\
  \bibnamefont {Berry}},\ }\bibfield  {title} {\bibinfo {title} {Optimal
  scaling quantum linear systems solver via discrete adiabatic theorem},\
  }\href@noop {} {\bibfield  {journal} {\bibinfo  {journal} {arXiv preprint
  arXiv:2111.08152}\ } (\bibinfo {year} {2021})}\BibitemShut {NoStop}%
\bibitem [{\citenamefont {Low}\ and\ \citenamefont
  {Chuang}(2019)}]{low2019hamiltonian}%
  \BibitemOpen
  \bibfield  {author} {\bibinfo {author} {\bibfnamefont {G.~H.}\ \bibnamefont
  {Low}}\ and\ \bibinfo {author} {\bibfnamefont {I.~L.}\ \bibnamefont
  {Chuang}},\ }\bibfield  {title} {\bibinfo {title} {Hamiltonian simulation by
  qubitization},\ }\href@noop {} {\bibfield  {journal} {\bibinfo  {journal}
  {Quantum}\ }\textbf {\bibinfo {volume} {3}},\ \bibinfo {pages} {163}
  (\bibinfo {year} {2019})}\BibitemShut {NoStop}%
\bibitem [{\citenamefont {Alase}\ \emph {et~al.}(2021)\citenamefont {Alase},
  \citenamefont {Nerem}, \citenamefont {Bagherimehrab}, \citenamefont
  {H{\o}yer},\ and\ \citenamefont {Sanders}}]{barrynew}%
  \BibitemOpen
  \bibfield  {author} {\bibinfo {author} {\bibfnamefont {A.}~\bibnamefont
  {Alase}}, \bibinfo {author} {\bibfnamefont {R.~R.}\ \bibnamefont {Nerem}},
  \bibinfo {author} {\bibfnamefont {M.}~\bibnamefont {Bagherimehrab}}, \bibinfo
  {author} {\bibfnamefont {P.}~\bibnamefont {H{\o}yer}},\ and\ \bibinfo
  {author} {\bibfnamefont {B.~C.}\ \bibnamefont {Sanders}},\ }\bibfield
  {title} {\bibinfo {title} {Tight bound for estimating expectation values from
  a system of linear equations},\ }\href@noop {} {\bibfield  {journal}
  {\bibinfo  {journal} {arXiv preprint arXiv:2111.10485}\ } (\bibinfo {year}
  {2021})}\BibitemShut {NoStop}%
\bibitem [{\citenamefont {Rall}(2020)}]{rall2020quantum}%
  \BibitemOpen
  \bibfield  {author} {\bibinfo {author} {\bibfnamefont {P.}~\bibnamefont
  {Rall}},\ }\bibfield  {title} {\bibinfo {title} {Quantum algorithms for
  estimating physical quantities using block encodings},\ }\href@noop {}
  {\bibfield  {journal} {\bibinfo  {journal} {Physical Review A}\ }\textbf
  {\bibinfo {volume} {102}},\ \bibinfo {pages} {022408} (\bibinfo {year}
  {2020})}\BibitemShut {NoStop}%
\bibitem [{\citenamefont {Chakraborty}\ \emph {et~al.}(2019)\citenamefont
  {Chakraborty}, \citenamefont {Gily{\'e}n},\ and\ \citenamefont
  {Jeffery}}]{chakraborty2019power}%
  \BibitemOpen
  \bibfield  {author} {\bibinfo {author} {\bibfnamefont {S.}~\bibnamefont
  {Chakraborty}}, \bibinfo {author} {\bibfnamefont {A.}~\bibnamefont
  {Gily{\'e}n}},\ and\ \bibinfo {author} {\bibfnamefont {S.}~\bibnamefont
  {Jeffery}},\ }\bibfield  {title} {\bibinfo {title} {The power of
  block-encoded matrix powers: Improved regression techniques via faster
  hamiltonian simulation},\ }\href@noop {} {\bibfield  {journal} {\bibinfo
  {journal} {Leibniz international proceedings in informatics}\ }\textbf
  {\bibinfo {volume} {132}} (\bibinfo {year} {2019})}\BibitemShut {NoStop}%
\bibitem [{\citenamefont {Linden}\ \emph {et~al.}(2020)\citenamefont {Linden},
  \citenamefont {Montanaro},\ and\ \citenamefont {Shao}}]{linden2020quantum}%
  \BibitemOpen
  \bibfield  {author} {\bibinfo {author} {\bibfnamefont {N.}~\bibnamefont
  {Linden}}, \bibinfo {author} {\bibfnamefont {A.}~\bibnamefont {Montanaro}},\
  and\ \bibinfo {author} {\bibfnamefont {C.}~\bibnamefont {Shao}},\ }\bibfield
  {title} {\bibinfo {title} {Quantum vs. classical algorithms for solving the
  heat equation},\ }\href@noop {} {\bibfield  {journal} {\bibinfo  {journal}
  {arXiv preprint arXiv:2004.06516}\ } (\bibinfo {year} {2020})}\BibitemShut
  {NoStop}%
\bibitem [{\citenamefont {Gleinig}\ and\ \citenamefont
  {Hoefler}(2021)}]{gleinig2021efficient}%
  \BibitemOpen
  \bibfield  {author} {\bibinfo {author} {\bibfnamefont {N.}~\bibnamefont
  {Gleinig}}\ and\ \bibinfo {author} {\bibfnamefont {T.}~\bibnamefont
  {Hoefler}},\ }\bibfield  {title} {\bibinfo {title} {An efficient algorithm
  for sparse quantum state preparation},\ }in\ \href@noop {} {\emph {\bibinfo
  {booktitle} {2021 58th ACM/IEEE Design Automation Conference (DAC)}}}\
  (\bibinfo {organization} {IEEE},\ \bibinfo {year} {2021})\ pp.\ \bibinfo
  {pages} {433--438}\BibitemShut {NoStop}%
\bibitem [{\citenamefont {Zhang}\ \emph {et~al.}(2022)\citenamefont {Zhang},
  \citenamefont {Li},\ and\ \citenamefont {Yuan}}]{zhang2022quantum}%
  \BibitemOpen
  \bibfield  {author} {\bibinfo {author} {\bibfnamefont {X.-M.}\ \bibnamefont
  {Zhang}}, \bibinfo {author} {\bibfnamefont {T.}~\bibnamefont {Li}},\ and\
  \bibinfo {author} {\bibfnamefont {X.}~\bibnamefont {Yuan}},\ }\bibfield
  {title} {\bibinfo {title} {Quantum state preparation with optimal circuit
  depth: Implementations and applications},\ }\href@noop {} {\bibfield
  {journal} {\bibinfo  {journal} {arXiv preprint arXiv:2201.11495}\ } (\bibinfo
  {year} {2022})}\BibitemShut {NoStop}%
\bibitem [{\citenamefont {Zhang}(2001)}]{zhang2001stochastic}%
  \BibitemOpen
  \bibfield  {author} {\bibinfo {author} {\bibfnamefont {D.}~\bibnamefont
  {Zhang}},\ }\href@noop {} {\emph {\bibinfo {title} {Stochastic methods for
  flow in porous media: coping with uncertainties}}}\ (\bibinfo  {publisher}
  {Elsevier},\ \bibinfo {year} {2001})\BibitemShut {NoStop}%
\bibitem [{\citenamefont {Jin}(2009)}]{Jin-Hyp}%
  \BibitemOpen
  \bibfield  {author} {\bibinfo {author} {\bibfnamefont {S.}~\bibnamefont
  {Jin}},\ }\bibfield  {title} {\bibinfo {title} {Numerical methods for
  hyperbolic systems with singular coefficients: well-balanced scheme,
  {H}amiltonian preservation, and beyond},\ }in\ \href
  {https://doi.org/10.1090/psapm/067.1/2605214} {\emph {\bibinfo {booktitle}
  {Hyperbolic problems: theory, numerics and applications}}},\ \bibinfo
  {series} {Proc. Sympos. Appl. Math.}, Vol.~\bibinfo {volume} {67}\ (\bibinfo
  {publisher} {Amer. Math. Soc., Providence, RI},\ \bibinfo {year} {2009})\
  pp.\ \bibinfo {pages} {93--104}\BibitemShut {NoStop}%
\bibitem [{\citenamefont {Lewis}\ and\ \citenamefont
  {Miller}(1984)}]{lewis1984computational}%
  \BibitemOpen
  \bibfield  {author} {\bibinfo {author} {\bibfnamefont {E.~E.}\ \bibnamefont
  {Lewis}}\ and\ \bibinfo {author} {\bibfnamefont {W.~F.}\ \bibnamefont
  {Miller}},\ }\href@noop {} {\emph {\bibinfo {title} {Computational methods of
  neutron transport}}}\ (\bibinfo  {publisher} {John Wiley and Sons, Inc., New
  York, NY},\ \bibinfo {year} {1984})\BibitemShut {NoStop}%
\bibitem [{\citenamefont {Fichtl}(2009)}]{fichtl2009stochastic}%
  \BibitemOpen
  \bibfield  {author} {\bibinfo {author} {\bibfnamefont {E.~D.}\ \bibnamefont
  {Fichtl}},\ }\emph {\bibinfo {title} {Stochastic methods for uncertainty
  quantification in radiation transport}},\ \href@noop {} {Ph.D. thesis},\
  \bibinfo  {school} {The University of New Mexico} (\bibinfo {year}
  {2009})\BibitemShut {NoStop}%
\bibitem [{\citenamefont {Mishra}\ \emph {et~al.}(2016)\citenamefont {Mishra},
  \citenamefont {Schwab},\ and\ \citenamefont {{\v{S}}ukys}}]{mishra2016multi}%
  \BibitemOpen
  \bibfield  {author} {\bibinfo {author} {\bibfnamefont {S.}~\bibnamefont
  {Mishra}}, \bibinfo {author} {\bibfnamefont {C.}~\bibnamefont {Schwab}},\
  and\ \bibinfo {author} {\bibfnamefont {J.}~\bibnamefont {{\v{S}}ukys}},\
  }\bibfield  {title} {\bibinfo {title} {Multi-level monte carlo finite volume
  methods for uncertainty quantification of acoustic wave propagation in random
  heterogeneous layered medium},\ }\href@noop {} {\bibfield  {journal}
  {\bibinfo  {journal} {Journal of Computational Physics}\ }\textbf {\bibinfo
  {volume} {312}},\ \bibinfo {pages} {192} (\bibinfo {year}
  {2016})}\BibitemShut {NoStop}%
\bibitem [{\citenamefont {Ryzhik}\ \emph {et~al.}(1996)\citenamefont {Ryzhik},
  \citenamefont {Papanicolaou},\ and\ \citenamefont
  {Keller}}]{ryzhik1996transport}%
  \BibitemOpen
  \bibfield  {author} {\bibinfo {author} {\bibfnamefont {L.}~\bibnamefont
  {Ryzhik}}, \bibinfo {author} {\bibfnamefont {G.}~\bibnamefont
  {Papanicolaou}},\ and\ \bibinfo {author} {\bibfnamefont {J.~B.}\ \bibnamefont
  {Keller}},\ }\bibfield  {title} {\bibinfo {title} {Transport equations for
  elastic and other waves in random media},\ }\href@noop {} {\bibfield
  {journal} {\bibinfo  {journal} {Wave motion}\ }\textbf {\bibinfo {volume}
  {24}},\ \bibinfo {pages} {327} (\bibinfo {year} {1996})}\BibitemShut
  {NoStop}%
\bibitem [{\citenamefont {Kuo}\ and\ \citenamefont
  {Shieh}(2022)}]{kuo2022boundary}%
  \BibitemOpen
  \bibfield  {author} {\bibinfo {author} {\bibfnamefont {Y.-C.}\ \bibnamefont
  {Kuo}}\ and\ \bibinfo {author} {\bibfnamefont {S.-F.}\ \bibnamefont
  {Shieh}},\ }\bibfield  {title} {\bibinfo {title} {Boundary effects on
  eigen-problems of discrete laplacian in lattices},\ }\href@noop {} {\bibfield
   {journal} {\bibinfo  {journal} {Taiwanese Journal of Mathematics}\ }\textbf
  {\bibinfo {volume} {1}},\ \bibinfo {pages} {1} (\bibinfo {year}
  {2022})}\BibitemShut {NoStop}%
\bibitem [{\citenamefont {Jin}\ \emph {et~al.}(2022{\natexlab{b}})\citenamefont
  {Jin}, \citenamefont {Liu},\ and\ \citenamefont {Yu}}]{JinLiuYu}%
  \BibitemOpen
  \bibfield  {author} {\bibinfo {author} {\bibfnamefont {S.}~\bibnamefont
  {Jin}}, \bibinfo {author} {\bibfnamefont {N.}~\bibnamefont {Liu}},\ and\
  \bibinfo {author} {\bibfnamefont {Y.}~\bibnamefont {Yu}},\ }\bibfield
  {title} {\bibinfo {title} {Time complexity analysis of quantum difference
  methods for linear high dimensional and multiscale partial differential
  equations},\ }\href@noop {} {\bibfield  {journal} {\bibinfo  {journal} {arXiv
  preprint arXiv:2202.04537}\ } (\bibinfo {year}
  {2022}{\natexlab{b}})}\BibitemShut {NoStop}%
\end{thebibliography}%

\end{document}